\documentclass[10pt,conference]{IEEEtran}

\usepackage{hyperref}
\usepackage{booktabs}
\usepackage{subcaption}
\usepackage{algorithm}
\usepackage[noend]{algpseudocode}
\usepackage{amsmath}
\usepackage{stmaryrd}
\usepackage{url}
\usepackage{listings}
\usepackage{graphicx}
\usepackage{wrapfig}
\usepackage{fancyvrb}
\usepackage{paralist}
\usepackage{caption}
\usepackage{comment}
\usepackage{xspace}
\usepackage{multirow}
\usepackage{booktabs}
\usepackage{enumerate}
\usepackage{enumitem}
\usepackage{caption}
\usepackage{graphicx}
\usepackage{listing}
\usepackage{tikz}
\usetikzlibrary{arrows.meta, positioning, calc, fit}
\usepackage{bussproofs}
\usepackage{circledsteps}
\usepackage{pgfplots}
\pgfplotsset{compat=newest}
\usepackage{amsthm}
\usepackage{bbding}
\usepackage{cite}
\usepackage{balance}

\newtheorem{definition}{Definition}
\newtheorem{theorem}{Theorem}

\newtheorem{example}{Example}

\newcommand{\irule}[2]%
   {\mkern-2mu\displaystyle\frac{#1}{\vphantom{,}#2}\mkern-2mu}
\newcommand{\irulelabel}[3]
{
\mkern-2mu
\begin{array}{ll}
\displaystyle\frac{#1}{\vphantom{,}#2} & \!\!\!\!\!\! #3
\end{array}
\mkern-2mu
}

\newcommand{\tool}{\textsc{Nosdaq}\xspace}
\newcommand{\eusolver}{\textsc{EUSolver}\xspace}
\newcommand{\morpheus}{\textsc{Morpheus}\xspace}
\newcommand{\concord}{\textsc{Concord}\xspace}
\newcommand{\neo}{\textsc{Neo}\xspace}

\newcommand{\ngds}{\textsc{NGDS}\xspace}

\newcommand{\blaze}{\textsc{Blaze}\xspace}

\newcommand{\flashExtract}{\textsc{FlashExtract}\xspace}

\newcommand{\hades}{\textsc{Hades}\xspace}
\newcommand{\mitra}{\textsc{Mitra}\xspace}
\newcommand{\scythe}{\textsc{Scythe}\xspace}

\newcommand{\patsql}{\textsc{PATSQL}\xspace}
\newcommand{\dynamite}{\textsc{Dynamite}\xspace}

\newcommand{\flashRelate}{\textsc{FlashRelate}\xspace}

\newcommand{\datamaran}{\textsc{Datamaran}\xspace}

\newcommand{\sqlizer}{\textsc{SQLizer}\xspace}
\newcommand{\sqlsynthesizer}{\textsc{SQLSynthesizer}\xspace}
\newcommand{\sickle}{\textsc{Sickle}\xspace}
\newcommand{\simpl}{\textsc{Simpl}\xspace}
\newcommand{\treex}{\textsc{Treex}\xspace}

\newcommand{\bfpara}[1]{\vspace{5pt} \noindent \textbf{\textit{#1}}}
\newcommand{\textbox}[1]{\vspace{5pt} \noindent \fbox{\parbox{.97\linewidth}{#1}}}

\colorlet{punct}{red!60!black}
\definecolor{background}{HTML}{EEEEEE}
\definecolor{delim}{RGB}{20,105,176}
\colorlet{numb}{magenta!60!black}
\definecolor{darkgray}{rgb}{.4,.4,.4}
\definecolor{purple}{rgb}{0.65, 0.12, 0.82}

\lstdefinelanguage{json}{
    basicstyle=\normalfont\ttfamily,
    numbers=left,
    numberstyle=\scriptsize,
    stepnumber=1,
    numbersep=8pt,
    showstringspaces=false,
    breaklines=true,
    literate=
     *{0}{{{\color{numb}0}}}{1}
      {1}{{{\color{numb}1}}}{1}
      {2}{{{\color{numb}2}}}{1}
      {3}{{{\color{numb}3}}}{1}
      {4}{{{\color{numb}4}}}{1}
      {5}{{{\color{numb}5}}}{1}
      {6}{{{\color{numb}6}}}{1}
      {7}{{{\color{numb}7}}}{1}
      {8}{{{\color{numb}8}}}{1}
      {9}{{{\color{numb}9}}}{1}
      {:}{{{\color{punct}{:}}}}{1}
      {,}{{{\color{punct}{,}}}}{1}
      {\{}{{{\color{delim}{\{}}}}{1}
      {\}}{{{\color{delim}{\}}}}}{1}
      {[}{{{\color{delim}{[}}}}{1}
      {]}{{{\color{delim}{]}}}}{1},
}

\lstdefinelanguage{MongoDB}{
  keywords={\$unwind, \$match, \$group, \$project, \$addFields, \$lookup, \$gt, \$lt, \$count, \$sum, \$min, \$max, \$avg},
  keywordstyle=\color{delim}\ttfamily,
  ndkeywords={},
  ndkeywordstyle=\color{darkgray}\bfseries,
  identifierstyle=\color{black},
  sensitive=false,
  comment=[l]{//},
  morecomment=[s]{/*}{*/},
  commentstyle=\color{purple}\ttfamily,
  stringstyle=\ttfamily,
  literate=
     *{0}{{{\color{numb}0}}}{1}
      {1}{{{\color{numb}1}}}{1}
      {2}{{{\color{numb}2}}}{1}
      {3}{{{\color{numb}3}}}{1}
      {4}{{{\color{numb}4}}}{1}
      {5}{{{\color{numb}5}}}{1}
      {6}{{{\color{numb}6}}}{1}
      {7}{{{\color{numb}7}}}{1}
      {8}{{{\color{numb}8}}}{1}
      {9}{{{\color{numb}9}}}{1}
      {:}{{{\color{punct}{:}}}}{1}
      {,}{{{\color{punct}{,}}}}{1}
      {.}{{{\color{punct}{.}}}}{1}
      {\{}{{{\color{delim}{\{}}}}{1}
      {\}}{{{\color{delim}{\}}}}}{1}
      {[}{{{\color{delim}{[}}}}{1}
      {]}{{{\color{delim}{]}}}}{1},
}

\newcommand{\set}[1]{\{#1\}}
\newcommand{\denot}[1]{\llbracket{#1}\rrbracket}
\newcommand{\pair}[2]{(#1, #2)}

\newcommand{\schema}{\mathcal{S}}
\newcommand{\db}{\mathcal{D}}
\newcommand{\exa}{\mathcal{E}}
\newcommand{\prog}{\mathcal{Q}}
\newcommand{\query}{\prog}

\newcommand{\abstype}{\mathcal{T}}
\newcommand{\wl}{\mathcal{W}}
\newcommand{\sketch}{\Omega}
\newcommand{\collection}{\mathcal{C}}
\newcommand{\attr}{a}
\newcommand{\name}{N}
\newcommand{\doc}{D}
\newcommand{\pred}{\phi}
\newcommand{\absSet}{\Lambda}
\newcommand{\matches}{\triangleleft}
\newcommand{\refines}{\sqsubseteq}
\newcommand{\dom}{\emph{dom}}
\newcommand{\hole}{\texttt{?}}

\newcommand{\type}{\tau}

\newcommand{\arrType}[1]{\textsl{Arr}\langle #1 \rangle}

\newcommand{\collectionType}{T_C}
\newcommand{\docType}{T_D}
\newcommand{\valueType}{T_V}
\newcommand{\primitiveType}{T_P}
\newcommand{\cmdType}{\textsl{Type}\xspace}

\newcommand{\cmdAny}{\textsl{Any}\xspace}
\newcommand{\cmdArr}{\textsl{Arr}\xspace}

\newcommand{\absColl}{\tilde{\collection}}
\newcommand{\augType}{\abstype}
\newcommand{\absDb}{\tilde{\db}}

\newcommand{\cmdProject}{\textsl{Project}\xspace}
\newcommand{\cmdMatch}{\textsl{Match}\xspace}
\newcommand{\cmdAddFields}{\textsl{AddFields}\xspace}
\newcommand{\cmdUnwind}{\textsl{Unwind}\xspace}
\newcommand{\cmdGroup}{\textsl{Group}\xspace}
\newcommand{\cmdLookup}{\textsl{Lookup}\xspace}
\newcommand{\cmdNull}{\textsl{Null}\xspace}
\newcommand{\cmdSum}{\textsl{Sum}\xspace}
\newcommand{\cmdAvg}{\textsl{Avg}\xspace}
\newcommand{\cmdMin}{\textsl{Min}\xspace}
\newcommand{\cmdMax}{\textsl{Max}\xspace}

\newcommand{\cmdSize}{\textsl{SizeEq}\xspace}
\newcommand{\cmdCount}{\textsl{Count}\xspace}
\newcommand{\cmdExists}{\textsl{Exists}\xspace}

\newcommand{\cmdNum}{\textsl{Num}\xspace}
\newcommand{\cmdString}{\textsl{String}\xspace}
\newcommand{\cmdBool}{\textsl{Bool}\xspace}
\newcommand{\cmdDatetime}{\textsl{Datetime}\xspace}
\newcommand{\cmdObjectId}{\textsl{ObjectId}\xspace}
\newcommand{\cmdValue}{\textsl{Value}\xspace}

\newcommand{\cmdMap}{\textsl{map}\xspace}

\newcommand{\cmdFilter}{\textsl{filter}\xspace}
\newcommand{\cmdIte}{\textsl{ite}\xspace}

\newcommand{\cmdExtractAttrs}{\textsl{ExtractAttrs}\xspace}
\newcommand{\cmdAddAttrs}{\textsl{AddAttrs}\xspace}
\newcommand{\cmdFlatMap}{\textsl{flatmap}\xspace}
\newcommand{\cmdFlattenArr}{\textsl{Flatten}\xspace}
\newcommand{\cmdDedup}{\textsl{dedup}\xspace}
\newcommand{\cmdGet}{\textsl{Get}\xspace}

\newcommand{\cmdHasAp}{\textsl{HasAp}\xspace}

\newcommand{\cmdAllNull}{\textsl{AllNull}\xspace}
\newcommand{\cmdId}{\textsl{Id}\xspace}

\newcommand{\cmdToDocType}{\textsl{ToDocType}\xspace}

\newcommand{\tmax}{\text{max}\xspace}
\newcommand{\tmin}{\text{min}\xspace}

\begin{document}

\title{Synthesizing Document Database Queries using Collection Abstractions}

\author{\IEEEauthorblockN{Qikang Liu, Yang He, Yanwen Cai, Byeongguk Kwak, Yuepeng Wang}
\IEEEauthorblockA{Simon Fraser University, Buranby, BC, Canada \\
\{qla116, yha244, yca452, bka47, yuepeng\}@sfu.ca}
}

\maketitle

\begin{abstract}
Document databases are increasingly popular in various applications, but their queries are challenging to write due to the flexible and complex data model underlying document databases. This paper presents a synthesis technique that aims to generate document database queries from input-output examples automatically. A new domain-specific language is designed to express a representative set of document database queries in an algebraic style. Furthermore, the synthesis technique leverages a novel abstraction of collections for deduction to efficiently prune the search space and quickly generate the target query. An evaluation of 110 benchmarks from various sources shows that the proposed technique can synthesize 108 benchmarks successfully. On average, the synthesizer can generate document database queries from a small number of input-output examples within tens of seconds.
\end{abstract}

\section{Introduction} \label{sec:intro}

Document databases like MongoDB~\cite{mongodb-web24} and CouchDB~\cite{couchdb-web24} have become increasingly popular in various real-world scenarios, such as online commercial platforms, financial services, gaming, and social media applications~\cite{when-mongo-web24}. Different from traditional relational databases that primarily use structured data like tables, document databases persist data in a semi-structured format such as JSON and BSON. While the semi-structured data format provides developers with great flexibility in storing and querying complex data structures directly, it also raises significant challenges for users to write queries for document databases.

To help users write document database queries in an easy and convenient fashion, we develop a synthesis technique to generate queries automatically. Inspired by prior work on automated synthesis of SQL queries for relational databases~\cite{SQLSynthesizer-ase2013,scythe-pldi2017,morpheus-pldi17}, our technique aims to generate document database queries from input-output examples. Specifically, the user only needs to provide a small number of examples to demonstrate the query, where the input example is a small document database consisting of a few documents, and the output example is the desired query result over the input. The goal of our synthesis technique is to generate a document database query such that executing the query over the input example produces the output example.

However, unlike synthesizing SQL queries, there are several key challenges to synthesizing queries for document databases.
\begin{itemize}[leftmargin=*]
\item \emph{Hierarchical and nested data structures}.
Document databases support hierarchical and nested data structures, such as arrays, documents, and their combinations. Since queries for document databases constantly operate over these complex data structures, it is crucial for synthesizers to reason about complex data structures efficiently for better performance.
\item\emph{Specialized query language}.
Query languages for document databases may use specialized operators over complex data structures that relational databases cannot handle. For instance, MongoDB uses a lookup operator in aggregation pipelines to query data over multiple collections. Synthesizers need to support an expressive query language for document databases while maintaining the efficiency of exploring a large search space of the target query.
\end{itemize}

To address these challenges, we have designed \emph{a new domain-specific language} based on the aggregation pipeline in MongoDB that can express a representative set of queries with core operators of document databases. The queries of this language are in an algebraic style similar to relational algebra but tailored towards document databases.

Furthermore, prior work on program synthesis proposed an approach to speed up the synthesis process by deduction~\cite{morpheus-pldi17,neo-pldi2018}. For fast synthesis of document database queries, we have adapted this approach to our setting and developed \emph{a novel abstraction for collections containing hierarchical and nested data structures} to prune the search space efficiently.
The key insight is that the ``shape'' and size of collections can help the synthesizer quickly prune incorrect queries, even if the query is partial. Thus, our abstraction consists of two pieces of information about the collection: First, it includes the \emph{type} of documents inside the collection. Second, it includes a logical formula describing constraints over the \emph{size} of the collection.

More specifically, our synthesis technique is presented schematically in Figure~\ref{fig:workflow}.
At a high level, the synthesis technique takes an iterative approach and has two phases in each iteration. In the first phase, the synthesizer aims to find a query sketch, which is a partial query with some unknown constructs. In the second phase, it tries to complete the sketch into a full query that can satisfy all provided input-output examples. In general, it is not efficient to check if a sketch is feasible to be completed into a correct full query by checking all possible completions against the examples, because a sketch may have a large number of completions. To avoid such inefficiency, the key part of our synthesis technique is a deduction engine, which can check if a sketch is feasible to get a correct query without checking its completions. In particular, the deduction engine can directly evaluate the sketch over \emph{abstractions of collections} and obtain an abstract collection. If the expected output example is a valid concretization of the resulting abstract collection, the synthesizer concludes the sketch is feasible to complete and proceeds to find a correct completion. Otherwise, the synthesizer can safely conclude the sketch is infeasible to complete, prune the search space accordingly, and propose a different sketch to the next iteration by refining the infeasible sketch.

\begin{figure}[!t]
\centering
\includegraphics[width=0.45\textwidth]{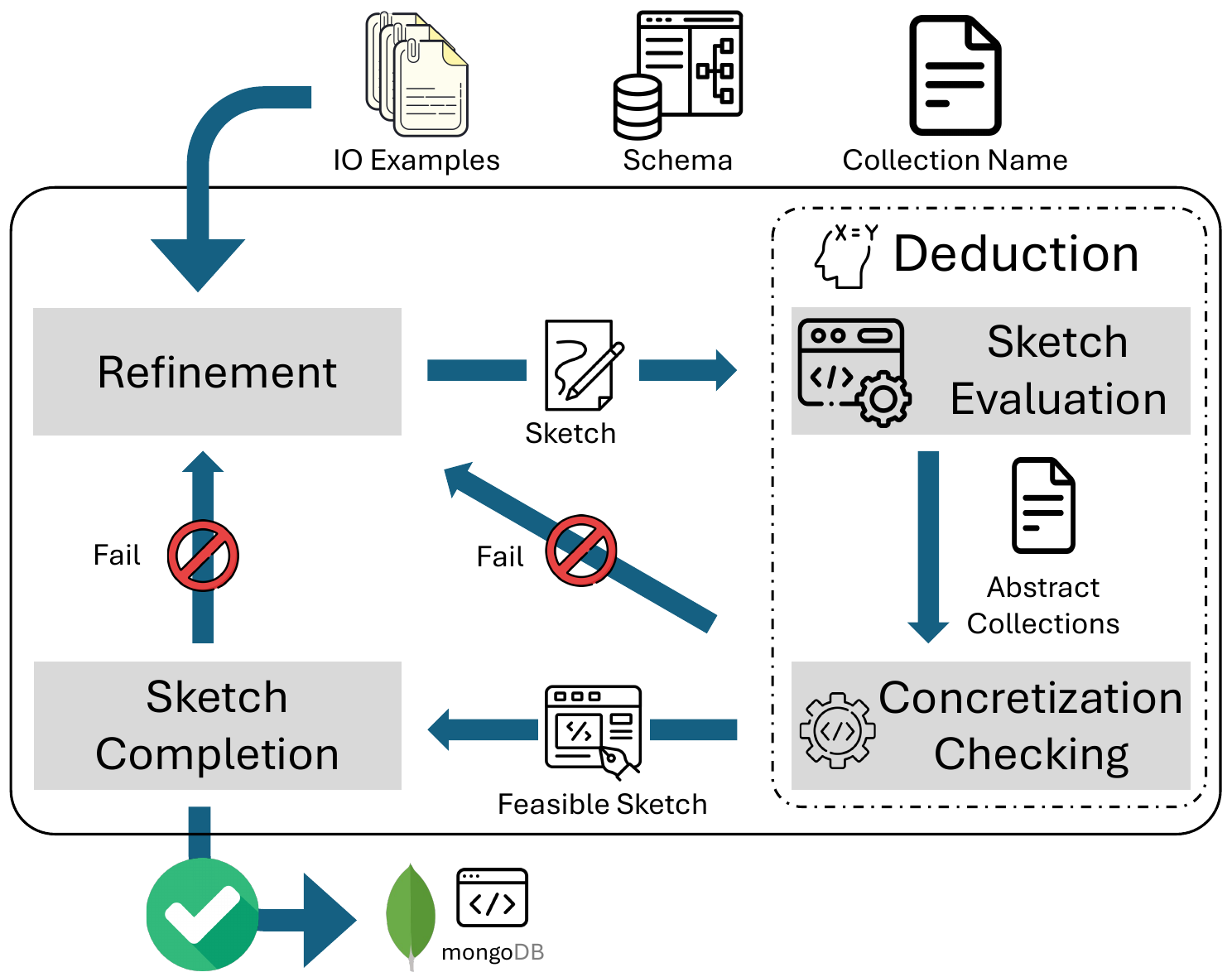}
\caption{Schematic workflow.}
\label{fig:workflow}
\vspace{-15pt}
\end{figure}

Based on this technique, we have developed a tool called \tool that can synthesize document database queries from input-output examples. To evaluate the synthesis technique, we have collected 110 benchmarks from various application scenarios, including StackOverflow, Kaggle, MongoDB official documents, and Twitter API documents. The evaluation result shows that \tool can successfully synthesize 108 document database queries within the 5-minute time limit. Furthermore, \tool only uses 1 -- 3 input-output examples and finishes query synthesis in an average of 25.4 seconds, which demonstrates the effectiveness and efficiency of our synthesis technique.

\bfpara{Contributions.}
To summarize, the main contributions of this paper are as follows.
\begin{enumerate}[leftmargin=*]
\item We develop a technique for synthesizing document database queries from input-output examples.
\item We design a new domain-specific language to express document database queries in algebraic style.
\item We design a novel abstraction for collections containing hierarchical and nested data structures and use this abstraction to speedup the synthesis of document database queries based on deduction.
\item We define the abstract semantics of document database queries based on our abstraction of collections.
\item We develop a tool called \tool and evaluate it over 110 benchmarks from various sources. The evaluation result shows that \tool is effective and efficient in synthesizing document database queries.
\end{enumerate}

\bfpara{Organization.}
The remainder of this paper is structured as follows. Section \ref{sec:overview} provides a motivating example to illustrate our technique. Section \ref{sec:formulation} formalizes the synthesis problem, and Section \ref{sec:abstraction} introduces collection abstractions. Sections \ref{sec:synthesis} and \ref{sec:impl} present the synthesis algorithm and its implementation details, respectively. Section \ref{sec:eval} presents the experimental setup and evaluation results. Section \ref{sec:related} discusses the related work, followed by a conclusion in Section \ref{sec:concl}.

\section{Motivating Example} \label{sec:overview}

To explain our synthesis technique, let us consider a concrete motivating example.
Given a document database collected from the Kaggle website that stores a list of Reddit posts. The database only has one collection called \texttt{posts} with the following schema
\footnote{
The database schema is simplified in this section for illustration.
}

\[
\footnotesize
\arrType{\set{\texttt{\_id}: \cmdString, \texttt{title}: \cmdString, \texttt{replies}: \arrType{\set{\texttt{depth}: \cmdNum}}}}
\]
where $\arrType{\type}$ denotes the array type of $\type$.
Specifically, the \texttt{posts} collection contains an array of documents, where each document has three attributes: \texttt{\_id}, \texttt{title}, and \texttt{replies}. The \texttt{replies} attribute is also an array of documents and the document has one attribute \texttt{depth} denoting the nesting level of the reply from the root post.

\begin{figure}[!t]
\footnotesize
\begin{lstlisting}[language=json, numbers=none]
{posts: [{
    _id: "1", title: "Title-1",
    replies: [{depth: 0}, {depth: 0}, {depth: 1}]
 }, {
    _id: "2", title: "Title-2",
    replies: [{depth: 0}, {depth: 1}, {depth: 2}]
 }, {
    _id: "3", title: "Title-3",
    replies: [{depth: 0}, {depth: 1}, {depth: 2},
        {depth: 3}]
 }]}
\end{lstlisting}
\vspace{-10pt}
\caption{Input example.}
\label{fig:example-input}
\vspace{-20pt}
\end{figure}

Now suppose the user wants to query the title of posts which have more than one non-zero-depth replies and the count of these replies.
\tool can help the user synthesize this query automatically. The user needs to provide small input-output examples to demonstrate their intention. For instance, Figure~\ref{fig:example-input} is an input example, and the corresponding output example is
\begin{center}
\begin{tabular}{c}
\begin{lstlisting}[language=json, numbers=none, basicstyle=\footnotesize\ttfamily]
[{reply_count: 3, title: "Title-3"},
 {reply_count: 2, title: "Title-2"}]
\end{lstlisting}
\end{tabular}
\end{center}

The goal of \tool is to synthesize a query such that executing the query on the input example produces the output example.
\tool takes an iterative approach to solve the synthesis problem. In each iteration, it first proposes a query sketch that may contain unknowns and then checks if the sketch is feasible to complete. If feasible, \tool completes the sketch into a full query by enumerative search and checks if any query satisfies the input-output example. If the sketch is infeasible to complete, then \tool refines the sketch and starts the next iteration.

\bfpara{First iteration.}
\tool starts with the simplest sketch in its domain-specific language -- the \texttt{posts} collection and checks its feasibility. To do so, the deduction engine of \tool uses the collection abstractions and evaluates the sketch based on its abstract semantics. 
Specifically, the abstraction for the \texttt{posts} collection is $\absColl = (\augType, \pred)$ where
\[
\footnotesize
\augType: \set{\texttt{\_id}: \cmdString, \texttt{title}: \cmdString, \texttt{replies}: \arrType{\set{\texttt{depth}: \cmdNum}}}
\]
is the type of inside documents and $\pred: l_0 = 3$ is the formula describing the size of the collection is 3.
\tool takes the abstract collection as input and evaluates the sketch \texttt{posts} based on the abstract semantics. The evaluation result is $\set{\absColl_1}$ where $\absColl_1 = (\augType, \pred)$, which is exactly the same as $\absColl$.
An important observation here is that the output example is not a concretization of $\absColl_1$ because its document has type  $\set{\texttt{title}: \cmdString, \texttt{reply\_count}: \cmdNum}$ and its size is 2.
Thus, \tool concludes the sketch \texttt{posts} is not feasible to complete to a correct query and starts to refine the sketch for the next iteration. In particular, \tool generates candidate sketches based on the grammar of its query language, such as $\cmdProject(\texttt{posts}, \vec{h})$ and $\cmdMatch(\texttt{posts}, \pred)$.

\bfpara{Deduction with collection abstractions.}
Several iterations later, \tool encounters the following sketch $\sketch_2$
\[
\small
\cmdProject(\cmdMatch(\cmdUnwind(\texttt{posts}, h_1) ,\pred), \vec{h_2})
\]
This time, the evaluation result is $\set{\absColl_2}$ where $\absColl_2 = (\augType_2, \pred_2)$ where $\augType_2$ is $\set{\texttt{title}: \cmdString}$ and $\pred_2$ is $l_0 = 3 \land l_1 \ge l_0 \land l_2 \le l_1 \land l_3 = l_2$, where $l_3$ corresponds to the size of $\absColl_2$.
The sketch $\sketch_2$ is still infeasible to complete, because the output document has an additional attribute \texttt{reply\_count} that does not match the type $\augType_2$.
\tool prunes this sketch $\sketch_2$ and continues to search for a feasible sketch.

\bfpara{Feasible sketch.}
After a few more iterations, \tool finds another sketch $\sketch_3$
\[
\small
\begin{array}{l}
\cmdProject(\cmdMatch(\cmdAddFields(\cmdGroup(\cmdMatch( \\
\hspace{2em} \cmdUnwind(\texttt{posts}, h_1), \pred), \vec{h_2}, \vec{a}, \vec{A}), \vec{h_3}, \vec{E}), \pred'), \vec{h_4}) \\
\end{array}
\]
The evaluation result of this sketch over the abstract semantics is a set of abstract collections $\absSet$, meaning the result can be some one inside $\absSet$. Among this set, there is an abstract collection $\absColl_3 = (\augType_3, \pred_3) \in \absSet$ where
\[
\footnotesize
\begin{array}{rcl}
\augType_3 &\hspace{-1em}:\hspace{-1em}& \set{\hole^+_0: \cmdAny, \hole^+_3: \cmdNum} \\
\pred_3 &\hspace{-1em}:\hspace{-1em}& l_0=3 \land l_1 \ge l_0 \land l_2 \le l_1 \land l_3 < l_2 \land l_4 = l_3 \land l_5 \le l_4 \land l_6 = l_5 \\
\end{array}
\]
Here, $\hole^+_0$ and $\hole^+_3$ denote placeholders that can match one or more attributes. $\cmdAny$ denotes any value type. $l_6$ is the variable that corresponds to the size of $\absColl_3$.
Observe that the output example is a concretization of abstract collection $\absColl_3$, because the \texttt{title} matches $\hole^+_0$ and \texttt{reply\_count} matches $\hole^+_3$. In addition, the size of the output collection is consistent with the size of $\absColl_3$, because $l_6 = 2 \land \pred_3$ is satisfiable. Therefore, \tool finds a feasible sketch $\sketch_3$.

\bfpara{Sketch completion.}
Given a feasible sketch $\sketch_3$, \tool aims to complete $\sketch_3$ by finding instantiations of all unknown operators in the sketch, such as $h_1, \vec{h_2}, \vec{a}$, etc. Towards this goal, \tool performs enumerative search and finds the following query finally
\[
\footnotesize
\begin{array}{l}
\cmdProject(\cmdMatch(\cmdAddFields(\cmdGroup(\cmdMatch( \\
\hspace{1em} \cmdUnwind(\texttt{posts}, \texttt{replies}), \texttt{replies.depth} > 0), \\
\hspace{1em} [\texttt{\_id}, \texttt{title}], [\texttt{reply\_count}], [\cmdCount()]), [\texttt{title}], \\
\hspace{1em} [\texttt{\_id.title}]), \texttt{reply\_count} > 1), [\texttt{reply\_count}, \texttt{title}]) \\
\end{array}
\]
Executing this query on the input example produces exactly the output example, so the synthesis process is finished. The query corresponds to the following MongoDB query
\begin{lstlisting}[language=MongoDB, numbers=none]
db.posts.aggregate([
  {$unwind: "$replies"},
  {$match: {"replies.depth": {$gt: 0}}},
  {$group:
    {_id: { _id: "$_id", title: "$title" },
     reply_count: { $count: {}}}},
  {$addFields: {title: "$_id.title"}},
  {$match: {reply_count: {$gt: 1}}},
  {$project: {_id: 0, reply_count: 1, title: 1}}])
\end{lstlisting}


\section{Problem Formulation} \label{sec:formulation}

In this section, we present formulations that are necessary for the rest of the paper and formally define our problem.

\subsection{Document Schema and Database}

We first precisely define the document schema and document database considered in this paper.

\begin{figure}[!t]
\centering
\small
\[
\begin{array}{rcl}
\text{Schema}~\schema & ::= & \set{\name_1 \mapsto T_{C_1}, \ldots, \name_m \mapsto T_{C_m}} \\
\text{Collection Type}~\collectionType & ::= & \arrType{\docType}\\
\text{Document Type}~\docType & ::= & \set{\attr_1: T_{V_1} \ldots \attr_n: T_{V_n}} \\
\text{Value Type}~\valueType & ::= & \docType ~|~ \arrType{\valueType}  ~|~  \primitiveType \\
\text{Primitive Type}~\primitiveType & ::= & \cmdNum ~|~ \cmdString ~|~ \cmdBool \\
    & ~|~ & \cmdDatetime ~|~ \cmdObjectId\\
\end{array}
\]
\[
\begin{array}{c}
N \in \textbf{Collection Names} \quad
\attr \in \textbf{Attributes} \\
\end{array}
\]
\vspace{-10pt}
\caption{Schema of document databases.}
\label{fig:schema-definition}
\vspace{-10pt}
\end{figure}

\bfpara{Document schema.}
As shown in Figure~\ref{fig:schema-definition}, a document schema $\schema$ is a map from collection names to collection types, where a collection type is an array of document types. A document type is a map from attributes to different value types, including document types, arrays, and primitive types such as \cmdNum, \cmdString, and \cmdBool.

\begin{figure}[!t]
\centering
\small
\[
\begin{array}{rcl}
\text{Database}~\db & ::= & \set{\name_1 \mapsto \collection_1, \ldots, \name_m \mapsto \collection_m} \\
\text{Collection}~\collection & ::= & [\doc] \\
\text{Document}~\doc & ::= & \set{\attr_1: v_1, \ldots, \attr_n: v_n} \\
\text{Value}~v & ::= & \doc ~|~ [v_1, \ldots, v_n] ~|~ c \\
\end{array}
\]
\[
N \in \textbf{Collection Names} \quad
\attr \in \textbf{Attributes} \quad
c \in \textbf{Constants}
\]
\vspace{-10pt}
\caption{Definition of document databases.}
\label{fig:db-definition}
\vspace{-10pt}
\end{figure}

\bfpara{Document database.}
As shown in Figure~\ref{fig:db-definition}, a document database is a map from collection names to their corresponding collections. A collection is an array of documents. A document is a map from attributes to values, where the value is a document, an array of values, or a constant of primitive types.

\begin{figure}[!t]
\footnotesize
\[
\begin{array}{c}

\irulelabel
{\begin{array}{c}
v \in \textsl{Constants} \quad
\cmdType(v) = \type \\
\end{array}}
{\vdash v : \type}
{\textrm{(T-Primitive)}}

\\ \ \\

\irulelabel
{\begin{array}{c}
\vdash v_i : \type \quad i = 1, \ldots, n \\
\end{array}}
{\vdash [v_1, \ldots, v_n] : \arrType{\type}}
{\textrm{(T-Array)}}

\\ \ \\

\irulelabel
{\begin{array}{c}
\doc = \set{\attr_1: v_1, \ldots, \attr_n: v_n} \\
\vdash v_i : \type_i \quad i = 1, \ldots, n \\
\end{array}}
{\vdash \doc : \set{\attr_1: \type_1, \ldots, \attr_n: \type_n}}
{\textrm{(T-Doc)}}

\\ \ \\



\irulelabel
{\begin{array}{c}
\db = \set{N_1 \mapsto \collection_1, \ldots, N_m \mapsto \collection_m} \\
\vdash \collection_i : \type_i \quad i = 1, \ldots, m \\
\end{array}}
{\vdash \db : \set{N_1 \mapsto \type_1, \ldots, N_m \mapsto \type_m}}
{\textrm{(T-DB)}}

\end{array}
\]
\vspace{-5pt}
\caption{Rules for conformance between databases and schemas.}
\label{fig:typing}
\vspace{-10pt}
\end{figure}

\bfpara{Typing and conformance.}
Figure~\ref{fig:typing} presents a set of typing rules for conformance checking between document databases and schemas, where judgments of the form $\vdash \db : \schema$ mean the database $\db$ conforms to schema $\schema$.
\footnote{
We view \cmdNull as a special value of any primitive type. If an attribute has both null values and non-null values in some collection, then its type will be the same as that of the non-null value.
}
Specifically, according to the \textrm{T-Primitive} rule, the type of a constant $v$ is simply $\cmdType(v)$.
The \textrm{T-Array} rule describes that all elements $v_i$ in an array must have the same type. If the element type is $\type$, then the array is of type $\arrType{\type}$.
The \textrm{T-Doc} rule states that the type of a document $\doc = \set{\attr_1: v_1, \ldots, \attr_n: v_n}$ is $\set{\attr_1: \type_1, \ldots, \attr_n: \type_n}$ where $\type_i$ is the type of $v_i$.
Finally, based on the \textrm{T-DB} rule, the schema (or the type) of a database is basically a map from collection names to types of their corresponding collections.


\subsection{Query Language}

Next, we describe the syntax and semantics
of our query language for document databases. The query language has a straightforward correspondence to a core query language of the MongoDB aggregation pipelines.

\begin{figure}[!t]
\centering
\small
\[
\hspace{-5pt}
\begin{array}{rcl}
\text{Query}~\prog & ::= & N ~|~ \cmdProject(\prog, \vec{h}) ~|~ \cmdMatch(\prog, \phi) \\ 
    & ~|~ & \cmdAddFields(\prog, \vec{h}, \vec{E}) ~|~  \cmdUnwind(\prog, h) \\
    & ~|~ & \cmdGroup(\prog, \vec{h}, \vec{a}, \vec{A}) ~|~ \cmdLookup(\prog, h, h, \name, a) \\
\text{Pred}~\pred & ::= & \top ~|~ \bot ~|~ h \odot c ~|~ \cmdSize(h, c) ~|~ \cmdExists(h) \\
    & ~|~ & \pred \land \pred ~|~ \pred \lor \pred ~|~ \neg \pred \\
\text{Expr}~E & ::= & h ~|~ h \oplus h ~|~ f(h) \\ 
\text{Agg}~A & ::= & \cmdSum(h) ~|~ \cmdAvg(h) ~|~ \cmdMin(h) ~|~ \cmdMax(h) ~|~ \cmdCount()\\
\text{LogicOp}~\odot & ::= & \leq ~|~ < ~|~ = ~|~ \neq ~|~ > ~|~ \geq \\
\text{ArithOp}~\oplus & ::= & + ~|~ - ~|~ \times ~|~  / ~|~ \%\\
\end{array}
\]
\[
\begin{array}{c}
\name \in \textbf{Collection Names} \quad
f \in \textbf{Math Functions} \\
c \in \textbf{Constants} \quad
\attr \in \textbf{Attributes} \quad
h \in \textbf{Access Paths} \\
\end{array}
\]
\vspace{-10pt}
\caption{Syntax of MongoDB Query. The two array parameters of AddFields must have the same length. The last two parameters of Group also must have the same length.}
\label{fig:syntax}
\vspace{-10pt}
\end{figure}

The syntax of the query language is shown in Figure~\ref{fig:syntax}. At a high level, a query is a sequence of operations including \cmdProject, \cmdMatch, \cmdAddFields, \cmdUnwind, \cmdGroup, and \cmdLookup, where different operators take different arguments such as a predicate $\pred$ or an expression $E$. Each operator corresponds to a stage of the MongoDB aggregation pipeline.
{
More specifically, the name $N$ simply retrieves collection $N$ from the database. $\cmdProject(\query, \vec{h})$ projects fields with access paths $\vec{h}$ from each document in the collection of $\query$.
$\cmdMatch(\query, \pred)$ filters the documents in $\query$'s collection, retaining only those satisfy the predicate $\pred$.
$\cmdAddFields(\query, \vec{h}, \vec{E})$ introduces new fields $\vec{h}$ with associated values of $\vec{E}$ to each document in $\query$. $\cmdUnwind(\query, h)$ deconstructs an array field $h$ in the documents of $\query$, mapping each document to a series of documents where the value of $h$ is replaced by individual elements of the original array. $\cmdGroup(\query, \vec{h}, \vec{a}, \vec{A})$ groups documents of $\query$ based on grouping keys $h$, transforming each group into a single document with new attributes $\vec{a}$ and aggregated values $\vec{A}$. Finally, $\cmdLookup(\query, h_1, h_2, N, a)$ adds a new attribute $a$ to each document of $\query$, where the attribute's value is a list of documents from a foreign collection $N$. This list only includes documents whose specified field $h_2$ in the foreign collection is the same as field $h_1$ in the original collection.
} 

The predicate $\pred$ can be true $\top$, false $\bot$, logical comparison $h \odot c$, size equality $\cmdSize(h, c)$, existence of an access path $\cmdExists(h)$, and boolean connectives. The expression $E$ can be an access path $h$, arithmetics $h \oplus h$, and mathematical functions $f(h)$. The access path is a sequence of attributes separated by dots such as $a_1.a_2.a_3$, denoting the path to access the data from the root document.



\begin{example}
Let us consider a document
\texttt{\{\_id: 1, name: "John", class: "SE", info: \{score: 90\}\}}.
The access path for the $\texttt{score}$ attribute in \texttt{info} is \texttt{info.score}.
\end{example}

\begin{example}
Consider a MongoDB query 
\begin{center}
\begin{tabular}{c}
\begin{lstlisting}[language=MongoDB, numbers=none, basicstyle=\small\ttfamily]
db.coll.aggregate([{$group:
    {_id: {name: "$name", class: "$class"}},
     total: {$sum: "$info.score"}}}])
\end{lstlisting}
\end{tabular}
\end{center}
It can be represented by the following query in our language
\[
\small
\begin{array}{l}
\cmdGroup(\texttt{coll}, [\texttt{name}, \texttt{class}], [\texttt{total}], [\cmdSum(\texttt{info.score})]) \\
\end{array}
\]
\end{example}

\begin{example}
Consider a collection 
\[
N = [\set{a:1, b:[2,3]}, \set{a:4, b:[5,6]}]
\]
The evaluation result of $\cmdUnwind(N, b)$ is
\[
[\set{a:1, b:2}, \set{a:1, b:3}, \set{a:4, b:5}, \set{a:4, b:6}]
\]
\end{example}



\subsection{Problem Statement}

Before we state the problem to solve in this paper, let us first define input-output examples.

\begin{definition}[Input-output example] \label{def:io-example}
An example $\exa$ over schema $\schema$ is a pair $(I, O)$ where $I$ is the document database over schema $\schema$ (i.e., $\vdash I : \schema$) and $O$ is the output collection.
\end{definition}

\bfpara{Synthesis problem.}
Given a database schema $\schema$, a collection name $\name \in \dom(\schema)$, and input-output examples $\vec{\exa}$ over $\schema$, the goal of our synthesis problem is to find a query $\query$ over collection $\name$ in the language shown in Figure~\ref{fig:syntax} such that for each example $(I, O) \in \vec{\exa}$, it holds that $\denot{\query}_I = O$. Here, $\denot{\query}_I$ represents the evaluation result of $\query$ given input database $I$. 

\section{Abstraction for Collections} \label{sec:abstraction}

In this section, we will introduce the abstraction for collections in document databases and how to compute abstractions for queries and sketches.

Intuitively, since collections in document databases contain an array of documents, the abstraction for collections should contain two pieces of information: (1) the \emph{type} of documents inside the collection and (2) the \emph{size} of the collection.
Based on this idea, we can define abstract collections and databases.

\begin{definition}[Abstract collection]
An abstract collection $\absColl = (\type, \pred)$ is a pair that consists of the type $\type$ of inside documents and the formula $\pred$ about the collection size.
\end{definition}

\begin{definition}[Abstract database] \label{def:abs-db}
An abstract database $\absDb = \set{\name_1 \mapsto \absColl_1, \ldots, \name_m \mapsto \absColl_m}$ is a map from collection names to abstract collections.
\end{definition}


Since the synthesis process also involves partial programs that may yield unknown attributes, values, or types in the documents, we now augment documents with a notion of placeholders.

\begin{definition}[Placeholder] \label{def:placeholder}
A placeholder $\hole^m$ denotes a top-level attribute that can match any concrete attribute and $m \in \set{1, +}$ denotes how many attributes it can match. $\hole^1$ means the placeholder matches exactly one attribute and $\hole^+$ means it can match one or more attributes.
\end{definition}

Accordingly, we update the type of documents with placeholders and augment attributes with a special type called \cmdAny that represents any possible value type.

\begin{definition}[Augmented type] \label{def:aug-type}
An augmented type $\augType$ is an extension of the document type $\docType$ in Figure~\ref{fig:schema-definition} where the attribute can be a named attribute or a placeholder and its type can be a value type $\valueType$ or $\cmdAny$ denoting any value type.
\end{definition}

\begin{example}
Let us consider an augmented type
\[
\small
\set{a: \cmdString, \hole_1^+: \cmdAny, \hole_2^+: \cmdNum, \hole_3^1: \arrType{\set{c: \cmdNum, d: \cmdString}}}
\]
Here, $\hole_1^+$ is a placeholder that matches one or more attributes of any type. $\hole_2^+$ is a placeholder that matches one or more attributes of \cmdNum type. $\hole_3^1$ is a placeholder that matches exactly one attribute corresponding to a collection where the document is of type $\set{c: \cmdNum, d: \cmdString}$.
\end{example}

Next, we can lift the notion of abstract collections to cases where placeholders are involved in the documents.

\begin{definition}[Abstract collection with placeholders] \label{def:abs-coll}
An abstract collection $\absColl = (\augType, \pred)$ is a pair consisting of (1) the augmented type $\augType$ of inside documents with potential placeholders and (2) the formula $\pred$ about the collection size.
\end{definition}

In the rest of the paper, we simply refer to abstract collections with placeholders as abstract collections, if the meaning is clear in the context.

\begin{definition}[Match] \label{def:match}
Let $\type$ be a document type and $\augType$ be an augmented type. We say $\type$ matches $\augType$, denoted $\type \matches \augType$, if
(1) replacing $\hole^1$ and $\hole^+$ with exactly one and at least one attributes respectively
and (2) replacing each occurrence of \cmdAny with a value type in $\augType$ yield a type equal to $\type$.
\end{definition}

Having defined the match relation between document types and augmented types, we can define the relation between concrete collections and abstract collections.

\begin{definition}[Collection concretization] \label{def:concretization}
A collection $\collection$ concretizes an abstract collection $\absColl = (\augType, \pred)$, denoted $\collection \refines \absColl$, if (1) $\type \matches \augType$ where $\vdash \collection : \arrType{\type}$ and (2) $\textsl{SAT}(\pred \land l_n = |\collection|)$ where $n = \textsl{MaxLabel}(\pred)$.
\end{definition}

Intuitively, if collection $\collection$ concretizes abstract collection $\absColl = (\augType, \pred)$, then (1) the type of documents in $\collection$ matches the augmented type $\augType$ of documents in $\absColl$; and (2) the size of $\collection$ is consistent with the size of $\absColl$ described by formula $\pred$.

\begin{example}
Consider the output collection $\collection$ in Section~\ref{sec:overview}
\begin{center}
\begin{tabular}{c}
\begin{lstlisting}[language=json, numbers=none, basicstyle=\small\ttfamily]
[{reply_count: 3, title: "Title-3"},
 {reply_count: 2, title: "Title-2"}]
\end{lstlisting}
\end{tabular}
\end{center}
Suppose $\absColl = (\augType, \pred)$ is an abstract collection where
\[
\footnotesize
\begin{array}{rcl}
\augType &\hspace{-1em}:\hspace{-1em}& \set{\hole_0^+: \cmdAny, \hole_2^+: \cmdNum} \\
\pred &\hspace{-1em}:\hspace{-1em}& l_0=3 \land l_1 \ge l_0 \land l_2 \le l_1 \land l_3 < l_2 \land l_4 = l_3 \land l_5 \le l_4 \land l_6 = l_5 \\
\end{array}
\]
Here, $l_6$ is the variable for the size of $\absColl$.
First, the type $\set{\texttt{reply\_count}: \cmdNum, \texttt{title}: \cmdString}$ matches the augmented type $\augType$.
Second, the size predicate $l_6 = 2$ is consistent with formula $\pred$.
Therefore, $\collection$ concretizes $\absColl$.
\end{example}


We can also lift the concretization relation to databases and abstract databases.

\begin{definition}[DB concretization] \label{def:db-concretization}
A database $\db$ over schema $\schema$ concretizes an abstract database $\absDb = \set{N_1 \mapsto \absColl_1, \ldots, N_m \mapsto \absColl_m}$, denoted $\db \refines \absDb$, if for all $1 \le i \le m$
\[
\schema[\name_i] = \arrType{\type_i} \Leftrightarrow \absColl_i = (\type_i, l_0 = |\db[\name_i]|)
\]
\end{definition}

\section{Synthesis using Collection Abstractions} \label{sec:synthesis}

In this section, we present our synthesis technique based on the abstraction of collections.

\subsection{High-Level Algorithm}

\begin{figure}[!t]
\small
\begin{algorithm}[H]
\caption{Synthesis Algorithm}
\label{algo:synthesis}
\begin{algorithmic}[1]
\Procedure{\textsc{Synthesize}}{$\schema, \name, \vec{\exa}$}
\vspace{2pt}
\Statex \textbf{Input:} Database schema $\schema$, collection name $\name$, input-output examples $\vec{\exa}$
\Statex \textbf{Output:} A query $\prog$ or $\bot$ indicating failure
\vspace{2pt}
\State $\wl \gets \{N\}$
\While{$\neg \textsl{IsEmpty}(\wl)$}
    \State $\sketch \gets \wl.\textsl{Dequeue}()$
    \If {$\textsc{Deduce}(\schema, \sketch, \vec{\exa})$}
        \State $\query \gets \textsc{CompleteSketch}(\schema, \sketch, \vec{\exa})$
        \If {$\query \ne \bot$}
            \Return $\query$
        \EndIf
    \EndIf

    \State $\wl.\textsl{EnqueueAll}(\textsc{Refine}(\sketch))$
    
\EndWhile
\State \Return $\bot$

\EndProcedure
\end{algorithmic}
\end{algorithm}
\vspace{-10pt}
\end{figure}

As shown in Algorithm~\ref{algo:synthesis}, our synthesis algorithm adapts the standard iterative approach based on worklists and sketches~\cite{morpheus-pldi17,neo-pldi2018} to the setting of document database queries.
Given a database schema $\schema$, a collection name $\name$, and input-output examples $\vec{\exa}$, the \textsc{Synthesize} procedure aims to find a query $\query$ over schema $\schema$ such that it satisfies the examples $\vec{\exa}$. 
Specifically, the worklist $\wl$ is initialized to be a singleton queue with the simplest sketch $\name$ (Line 2).
While the worklist is not empty, the synthesis procedure enters a loop (Lines 3 -- 8) that dequeues the current sketch $\sketch$ (Line 4) and checks if it is feasible to complete (Line 5). If yes, the procedure invokes the \textsc{CompleteSketch} procedure and tries to obtain a correct query (Lines 6--7). If the sketch is infeasible to complete or all of its completions are incorrect, the procedure also invokes the \textsc{Refine} procedure to transform the current sketch $\sketch$ to a set of sketches based on the grammar in Figure~\ref{fig:syntax} (Line 8). This synthesis procedure is repeated until a correct query $\query$ is found (Line 7) or returns $\bot$ if the worklist is empty.

\subsection{Sketch Enumeration and Refinement}

\begin{definition}[Sketch] \label{def:sketch}
A sketch $\sketch$ is a query $\query$ where only the collection name is known and other arguments are unknown.
\end{definition}

\begin{example}
Let us consider again the following sketch from the motivating example.
\[
\cmdProject(\cmdMatch(\cmdUnwind(\texttt{posts}, h_1), \pred), \vec{h_2})
\]
Here, the collection name \texttt{posts} is known, but access path $h_1$, predicate $\pred$, and access paths $\vec{h_2}$ are unknown.
\end{example}

Given a sketch $\sketch$ over collection $\name$, the \textsc{Refine} procedure substitutes the collection $\name$ with all possible query operators shown in Figure~\ref{fig:syntax} and obtains a set $S_{\sketch} = \set{\cmdProject(N, \vec{h}), \cmdMatch(N, \phi), \cmdAddFields(N, \vec{h}, \vec{E}), \\ \cmdUnwind(N, h), \cmdGroup(N, \vec{h}, \vec{a}, \vec{A}), \cmdLookup(N, h, h, \name, a)}$ and produces six new sketches. The refined sketches are
\[
\set{\sketch[\sketch_i/\name] ~|~ \sketch_i \in S_{\sketch}}
\]

\subsection{Abstract Semantics}

Since the key novelty of our synthesis technique is performing deduction on collection abstractions to prune infeasible sketches, we first introduce the abstract semantics of executing sketches over abstract collections.

\begin{figure}[!t]
\footnotesize
\[
\hspace{-10pt}
\begin{array}{c}

\irulelabel
{\begin{array}{c}
\absColl = \absDb[N] \\
\end{array}}
{\absDb, \type_O \vdash N \Downarrow \set{\absColl}}
{\textrm{(A-Collection)}}

\\ \ \\

\irulelabel
{\begin{array}{c}
\absDb, \type_O \vdash \sketch \Downarrow \absSet \quad
\pair{\abstype}{\pred} \in \absSet \\
\pair{\abstype}{\pred \land l_j \le l_i} \in \absSet' \\
\cmdId(\sketch) = i \quad
\cmdId(\cmdMatch(\sketch, P)) = j \\
\end{array}}
{\absDb, \type_O \vdash \cmdMatch(\sketch, P) \Downarrow \absSet'}
{\textrm{(A-Match)}}

\\ \ \\

\irulelabel
{\begin{array}{c}
\absDb, \type_O \vdash \sketch \Downarrow \absSet \quad
\pair{\augType}{\pred} \in \absSet \\
\type_k = \cmdToDocType(\augType)\\
\pair{(\augType - \type_k) \cup (\type_k \cap \type_O)}{\pred \land l_j = l_i} \in \absSet' \\
\cmdId(\sketch) = i \quad
\cmdId(\cmdProject(\sketch, \vec{h})) = j \\
\end{array}}
{\absDb, \type_O \vdash \cmdProject(\sketch, \vec{h}) \Downarrow \absSet'}
{\textrm{(A-Project)}}

\\\\\

\irulelabel
{\begin{array}{c}
\absDb, \type_O \vdash \sketch \Downarrow \absSet \quad
\pair{\augType}{\pred} \in \absSet \\
\pair{\augType \cup $\set{$\texttt{?}_0^+$: \cmdAny}$}{\pred \land l_j = l_i} \in \absSet' \\
\cmdId(\sketch) = i \quad
\cmdId(\cmdAddFields(\sketch, \vec{h}, \vec{E})) = j \\
\end{array}}
{\absDb, \type_O \vdash \cmdAddFields(\sketch, \vec{h}, \vec{E}) \Downarrow \absSet'}
{\textrm{(A-AddFields)}}

\\\\\

\irulelabel
{\begin{array}{c}
\absDb, \type_O \vdash \sketch \Downarrow \absSet \quad
\pair{\augType}{\pred} \in \absSet \\
\textsl{Type}(a_A) = \arrType{\tau} \land \textsl{NotInArr}(a_A)\\
\set{\pair{\augType[\type/a_A]}{\pred \land l_j \ge l_i}
 | a_A \in \augType \land \forall p.\forall q. a_A \ne \texttt{?}_p^q}  \subseteq \absSet' \\
\cmdId(\sketch) = i \quad
\cmdId( \cmdUnwind(\sketch, h)) = j \\
\end{array}}
{\absDb, \type_O \vdash \cmdUnwind(\sketch, h) \Downarrow \absSet'}
{\textrm{(A-Unwind)}}

\\\\\

\irulelabel
{\begin{array}{c}
\absDb, \type_O \vdash \sketch \Downarrow \absSet \quad
\pair{\augType}{\pred} \in \absSet \\
F = \set{\textsl{ToDocType}(\absDb[N]_\augType) | N \in \textsl{dom}(\absDb)}\\
\set{\pair{\augType \cup $\set{$\texttt{?}_j^1$:$\arrType{\tau_F}$}$}{\pred \land l_j = l_i} | \tau_F \in F}  \subseteq \absSet' \\
\cmdId(\sketch) = i \quad
\cmdId(\cmdLookup(\sketch, h, h, N, a)) = j \\
\end{array}}
{\absDb, \type_O \vdash \cmdLookup(\sketch, h, h, N, a) \Downarrow \absSet'}
{\textrm{(A-Lookup)}}

\\\\\

\irulelabel
{\begin{array}{c}
\absDb, \type_O \vdash \sketch \Downarrow \absSet \quad
\pair{\augType}{\pred} \in \absSet \\
G = $\set{\set{$\texttt{?}_j^+$: \cmdNum}, \set{}}$\\
\set{\pair{\texttt{\set{\_id:$\tau_K$}} \cup \type_g}{\pred \land l_j < l_i} \\
| \tau_K \subseteq \cmdToDocType(\augType) \land \type_g \in G} \subseteq \absSet' \\
\cmdId(\sketch) = i \quad
\cmdId(\cmdGroup(\sketch, \vec{h}, \vec{a}, \vec{A}) = j \\
\end{array}}
{\absDb, \type_O \vdash \cmdGroup(\sketch, \vec{h}, \vec{a}, \vec{A}) \Downarrow \absSet'}
{\textrm{(A-Group)}}

\end{array}
\]
\vspace{-5pt}
\caption{Abstract Semantics.  The \cmdToDocType function transforms an augmented type to a document type by deleting all placeholder attributes and the attributes with \cmdAny type. The \textsl{NotInArr} checks whether an attribute is not nested in an array type otherwise it is unable to be unwinded.}
\label{fig:abstract-semantics}
\vspace{-10pt}
\end{figure}

{
At a high level, the abstract semantics is consistent with the concrete semantics in describing how an operator modifies the collection size and the type of its documents, but it applies to the abstract database.
}
Formally, the abstract semantics is defined in Figure~\ref{fig:abstract-semantics}, where judgments of the form $\absDb, \type_O \vdash \sketch \Downarrow \absSet$ mean that a sketch $\sketch$ evaluates to a set of abstract collections $\absSet$ given an abstract database $\absDb$ and the output document type $\type_O$.
{
Specifically, by the \textrm{A-Collection} rule, the only abstract collection for query $\name$ can be obtained by looking up the abstract database $\absDb$.
By \textrm{A-Match}, $\cmdMatch$ reduces the collection size without changing the type of its inside documents.
By \textrm{A-Project}, $\cmdProject$ preserves the collection size but modifies the document type. In particular, the output document only retains a subset of the original attributes, and the remaining attributes can be inferred from the output.
According to \textrm{A-AddFields}, $\cmdAddFields$ adds one or more attributes of $\cmdAny$ type without changing the size of the collection.
By the \textrm{A-Unwind} rule, $\cmdUnwind(\sketch, h)$ potentially increases the collection size, deconstructs the array $h$ of sketch $\sketch$, and updates the type accordingly.
By the \textrm{A-Lookup} rule, $\cmdLookup$ preserves the collection size but introduces a new attribute to the $\augType$, where the type of the new attribute is the same as that of the foreign collection.
Finally, as shown in the \textrm{A-Group} rule, $\cmdGroup$ reduces the collection size and constructs a new type. In particular, it introduces a new attribute \texttt{\_id} as the key and uses a new $\type_g$ to represent a series of numeric attributes for aggregation results. $\type_g$ can also be empty, indicating the absence of aggregation attributes.
}

\begin{example}
Consider again the following sketch in Section~\ref{sec:overview}
\[
\cmdProject(\cmdMatch(\cmdUnwind(\texttt{posts}, h_1),\pred), \vec{h_2})
\]
Based on the rules in Figure~\ref{fig:ablation-study}, we recursively evaluate the sketch. The evaluation result of \texttt{posts} is 
\[
\begin{array}{l}
\set{(\set{\texttt{\_id}: \cmdString, \texttt{title}: \cmdString, \texttt{replies}:\\ \qquad \arrType{\set{\texttt{depth}: \cmdNum}}}, l_0=3)}
\end{array}
\]
The result of $\cmdUnwind(\texttt{posts}, h_1)$ is 
\[
\begin{array}{l}
\set{(\set{\texttt{\_id}: \cmdString, \texttt{title}: \cmdString, \texttt{replies}:\\ \qquad \set{\texttt{depth}: \cmdNum}}, l_0=3 \land l_1 \ge l_0)}
\end{array}
\]
The result of $\cmdMatch(\cmdUnwind(\texttt{posts}, h_1), \pred)$ is 
\[
\begin{array}{l}
\set{(\set{\texttt{\_id}: \cmdString, \texttt{title}: \cmdString, \texttt{replies}:\\ \qquad \set{\texttt{depth}: \cmdNum}}, l_0=3 \land l_1 \ge l_0 \land l_2 \le l_1)}
\end{array}
\]
The result of $\cmdProject(\cmdMatch(\cmdUnwind(\texttt{posts}, h_1), \pred), \vec{h_2}$) is 
\[
\begin{array}{l}
\set{(\set{\texttt{title}: \cmdString}, l_0=3 \land l_1 \ge l_0 \land l_2 \le l_1 \land l_3 = l_2)}
\end{array}
\]
\end{example}

Next, we establish the relationship among queries, sketches, concrete semantics, and abstract semantics with a theorem.

\begin{theorem} \label{thm:abs-soundness}
Let $\absDb$ be an abstract database over schema $\schema$, $\sketch$ be a sketch, $\query$ be a query that is a completion of $\sketch$, and $(I, O)$ be an input-output example, where $\vdash I: \schema$ and $\vdash O : \arrType{\type_O}$.
If $\denot{\query}_{I} = O$, $I \refines \absDb$, and $\absDb, \type_O \vdash \sketch \Downarrow \absSet$, then there exists an abstract collection $\absColl \in \absSet$ such that $O \refines \absColl$.
\end{theorem}

{
Intuitively, the theorem states that the abstract semantics is correct with respect to the concrete semantics. In particular, if the input is a concretization of the abstract database and the query is a completion of the sketch, then the evaluation result of the sketch on the abstract database is an over-approximation of the output produced by executing the query on the input.
}

\subsection{Deduction by Collection Abstractions}

Next, let us present how to perform deduction based on the collection abstractions.

\bfpara{Deduction algorithm.}
Our deduction algorithm is shown in Algorithm~\ref{algo:deduce}. For each example $\exa_j = (I_j, O_j)$, we compute the document type in $O_j$. The \textsc{ComputeType} computes the type of $O_j$ by the typing rules in Figure~\ref{fig:typing} (Line 4). Then the \textsl{In} function extracts the document type from the type of $O_j$, namely $\textsl{In}(\arrType{\type}) = \type$. We also compute the abstract input database $\absDb_j$ by computing all the abstractions of collections in the database (Line 5). Each collection name $N_i$ is mapped to an abstract collection whose augmented type is the document's type inside the collection and the predicate is $l_0$ equals the collection size.
For all pairs of abstract input database $\absDb_j$ and output document type $\type_{Oj}$, we evaluate the sketch $\sketch$ based on the abstract semantics in Figure~\ref{fig:abstract-semantics} and get a set of abstract collections for each example (Line 6).
If for each example $(I_j, O_j)$, there is an abstract collection $\absColl \in \absSet_j$ such that $O_j$ is a concretization of $\absColl$, then the sketch is feasible to complete (Line 7). Otherwise, the sketch is infeasible.

\begin{figure}[!t]
\small
\begin{algorithm}[H]
\caption{Deduction by Abstract Collections}
\label{algo:deduce}
\begin{algorithmic}[1]
\Procedure{\textsc{Deduce}}{$\schema, \sketch, \vec{\exa}$}
\vspace{2pt}
\Statex \textbf{Input:} The database schema $\schema$, a sketch $\sketch$ and input-output examples $\vec{\exa}$
\Statex \textbf{Output:} $\top$ if $\sketch$ is feasible otherwise $\bot$

\vspace{2pt}
\For {$j \gets 1 ~\text{to}~ |\vec{\exa}|$}
    \State $(I_j, O_j) \gets \exa_j$
    \State $\type_{Oj} \gets \textsl{In}(\textsc{ComputeType}(O_j))$
    \State $\absDb_j \gets \set{{N_i \mapsto (\textsl{In}(\schema[N_i]), l_0 = |I_j[N_i])| N_i \in \textsl{dom}(\schema))}}$
    \State $\absSet_j \gets \textsc{Eval}(\absDb_j, \type_{Oj}, \sketch)$
\EndFor
\If {$\forall j. \exists \absColl. \absColl \in \absSet_j \land O_j \refines \absColl$}
    \Return $\top$
\Else
    ~\Return $\bot$
\EndIf
\EndProcedure
\end{algorithmic}
\end{algorithm}
\vspace{-10pt}
\end{figure}


\bfpara{Concretization check.}
Recall from Definition~\ref{def:concretization} that to check a collection $\collection$ is a concretization of abstract collection $\absColl = (\augType, \pred)$, we need to check (1) the type $\type$ of documents inside $\collection$ matches $\augType$, i.e. $\type \matches \augType$, and (2) the size of $\collection$ is consistent with the formula $\pred$. We use an off-the-shelf SMT solver to check condition (2) by checking the satisfiability of formula $\pred \land l_n = |\collection|$ where $n = \textsl{MaxLabel}(\pred)$. We also develop a procedure for type match based on Definition~\ref{def:match}, which can be best explained with the following example.

\begin{example}\label{example:abstract-type}
Suppose we have an augmented type
\[
\begin{array}{rcl}
\augType &\hspace{-2em}=\hspace{-2em}& \set{\texttt{name}: \cmdString, \texttt{id}: \cmdString, \texttt{info}: \set{\texttt{tel}: \cmdString}, \\
&& \texttt{?}_1^+: \cmdNum, \texttt{?}_2^+: \cmdAny, \\
&& \texttt{?}_3^1: \cmdArr \langle \set{\texttt{profId}: \cmdString, \texttt{profName}: \cmdString} \rangle }
\end{array}
\]
and document type $\type$
\[
\begin{array}{ll}
\set{ \hspace{-2em}& \texttt{id}: \cmdString, \texttt{name}: \cmdString, \texttt{info}: \set{tel: \cmdString}, \\
& \texttt{newField}: \cmdBool, \texttt{sum}: \cmdNum, \\
& \texttt{profs}: \cmdArr \langle \set{\texttt{profId}: \cmdString, \texttt{profName}: \cmdString} \rangle }
\end{array}
\]
Here, $\set{\texttt{name}: \cmdString, \texttt{id}: \cmdString, \texttt{info}: \set{\texttt{tel}: \cmdString}}$ in $\augType$ is matched by $\set{\texttt{name}: \cmdString, \texttt{id}: \cmdString, \texttt{info}: \set{\texttt{tel}: \cmdString}}$ in $\type$, because the corresponding attributes have the same names and types.
$\set{\texttt{?}_3^1: \texttt{Arr} \langle \set{\texttt{profId}: \cmdString, \texttt{profName}: \cmdString} \rangle }$ is matched by $\set{\texttt{profs}: \texttt{Arr} \langle \set{\texttt{profId}: \cmdString, \texttt{profName}: \cmdString} \rangle }$, because 
\texttt{profs} has the same type as placeholder $\hole^1_3$ and $\hole^1_3$ matches exactly one attribute.
Finally, $\set{\hole^+_1: \cmdNum}$ is matched by $\set{\texttt{sum}: \cmdNum}$ because they have the same type, and $\set{\hole^+_2: \cmdAny}$ is matched by $\set{\texttt{newField}: \cmdBool}$ because \cmdAny can match any value type.
\end{example}



To understand why our deduction algorithm is correct, let us consider the following theorem.

\begin{theorem} \label{thm:deduce-soundness}
Given a database schema $\schema$, a sketch $\sketch$, and input-output examples $\vec{\exa}$, if $\textsc{Deduce}(\schema, \sketch, \vec{\exa})$ returns $\bot$, then there is no completion $\query$ of $\sketch$ such that for all $(I, O) \in \vec{\exa}$, $\denot{\query}_I = O$.
\end{theorem}

{
Intuitively, the theorem states that our deduction-based pruning is sound. In other words, if the deduction algorithm returns $\bot$ for a sketch, then no completions of the sketch satisfy all the input-output examples.
}



\subsection{Sketch Completion}
The \textsc{CompleteSketch} takes as input a schema $\schema$, a sketch $\sketch$, and input-output examples $\vec{\exa}$ and returns a query $\query$ satisfying all examples or $\bot$ if such a query does not exist.  We use an enumerative search algorithm to fill unknowns in the sketch according to the query operators.
\begin{enumerate}[leftmargin=*]
\item \cmdProject.
We compute the common attributes in the input and output and use these common attributes as arguments.
\item \cmdMatch.
We enumerate all predicates obtained from a combination of access paths, constants, comparisons, and logic connectives. Also, the observational equivalent class is used to avoid duplicate predicates.
\item \cmdAddFields.
We enumerate all possible expressions for newly generated attributes.
\item \cmdUnwind. We enumerate all array attributes in the top level of the document and unwind them.
\item \cmdGroup. 
We enumerate all group keys and accumulators and use value-based analysis to prune impossible accumulators.
\item \cmdLookup. We enumerate all foreign collections and their attributes as arguments.
\end{enumerate}

In addition, we also perform type checking to prune impossible arguments. For instance, if the value for a newly generated attribute has a different type than it should be in the output, we prune this completion from the search space.


{
We now conclude this section with two theorems about the overall synthesis algorithm.
}

\begin{theorem}[Soundness] \label{thm:soundness}
Let $\schema$ be a database schema, $\vec{\exa}$ be input-output examples, and $N$ be a collection name. Suppose \textsc{CompleteSketch} is sound, if $\textsc{Synthesize}(\schema, \vec{\exa}, N)$ returns a query $\query$, then $\query$ satisfies examples $\vec{\exa}$.
\end{theorem}

\begin{theorem}[Completeness] \label{thm:completeness}
Let $\schema$ be a database schema, $\vec{\exa}$ be input-output examples, and $N$ be a collection name. Suppose \textsc{CompleteSketch} is complete, if there exists a query accepted by the grammar in Figure~\ref{fig:syntax} that is over collection $N$ and satisfies examples $\vec{\exa}$, then $\textsc{Synthesize}(\schema, \vec{\exa}, N)$ does not return $\bot$.
\end{theorem}

{
Intuitively, the soundness theorem states that if the synthesis algorithm returns a query, then the query satisfies all input-output examples. The completeness theorem ensures that if there exists a query in our language satisfying all input-output examples, then the synthesis algorithm can find a query.
}

\section{Implementation} \label{sec:impl}
We have implemented the proposed synthesis technique in a tool called \tool and use Z3~\cite{z3-tacas08} as the SMT solver.

\bfpara{Heuristics for sketch completion.}
Based on the observation that most \cmdGroup operators do not have more than two group keys, we limit the number of group keys to two during sketch completion.
In addition, although \tool supports simple constants (e.g., null) in sketch completion, it expects the user to provide more complicated constants such as string literals.



\bfpara{Translation to MongoDB queries.}
\tool performs syntax-directed translation to transform the document database query in its domain-specific language to the MongoDB query language. Furthermore, it also performs optimizations to improve the conciseness and efficiency of translated queries, such as merging continuous \cmdAddFields and \cmdProject operators.


\section{Evaluation} \label{sec:eval}

In this section, we present several experiments that are designed to answer the following research questions.
\begin{enumerate}[leftmargin=3em]
\item[\textbf{RQ1.}]
Is \tool effective and efficient to synthesize document database queries from input-output examples?
\item[\textbf{RQ2.}]
How does each component of the collection abstraction affect synthesis time?
\item[\textbf{RQ3.}]
How does \tool compare against other baseline synthesizers?
\item[\textbf{RQ4.}]
How does the collection size of input-output examples impact the performance of \tool?
\end{enumerate}

\bfpara{Experimental setup.}
All experiments are conducted on a machine with an Intel i9-13905H CPU and 32 GB of physical memory, running the Ubuntu 22.04 WSL2 operating system.

\subsection{Benchmarks}

\begin{table}[!t]
\caption{Statistics of datasets. \textbf{\#n} is the number of benchmarks. \textbf{\#a}, \textbf{\#d}, \textbf{\#e}, \textbf{\#i}, \textbf{\#o}, \textbf{\#c} denote the average number of document attributes, document depths, examples, collection sizes in input and output examples, and constants, respectively.}
\label{table:benchmark}
\centering
\begin{tabular}{|c|c|c|c|c|c|c|c|}
\hline
\textbf{dataset} &\textbf{\#n} &  \textbf{\#a} & \textbf{\#d} & \textbf{\#e} & \textbf{\#i} & \textbf{\#o} &\textbf{\#c}\\
\hline
StackOverflow & 33 & 4.9 & 1.5 & 2.5 & 2.4 & 1.4 & 0.9\\
\hline
MongoDB Document & 26 & 5.4 & 1.4 & 1.1 & 4.7 & 2.6 & 0.6\\
\hline
Twitter API & 5 & 18.4 & 2.6 & 1.0 & 2.0 & 2.6 & 0.0\\
\hline
Kaggle & 46 & 19.8 & 4.1 & 1.0 & 1.8 & 3.6 & 0.5\\
\hline
\textbf{Total} & 110 & 11.9 & 2.6 & 1.5 & 2.6 & 2.3 & 0.6\\
\hline
\end{tabular}
\end{table}

\begin{table}[!t]
\caption{Statistics of ground truth queries. \textbf{\#s}, \textbf{\#op}, \textbf{\#P}, \textbf{\#M}, \textbf{\#L}, \textbf{\#U}, \textbf{\#G}, \textbf{\#A} denote the number of AST nodes, query operators, \cmdProject, \cmdMatch, \cmdLookup, \cmdUnwind, \cmdGroup, \cmdAddFields, respectively.}
\label{table:ground-truth-query}
\centering
    \begin{tabular}{|p{4em}|p{1.5em}|p{1.5em}|p{1.5em}|p{1.5em}|p{1.5em}|p{1.5em}|p{1.5em}|p{1.5em}|p{1.5em}|}
        \hline
        \textbf{dataset} & & \textbf{\#s} & \textbf{\#op} & \textbf{\#P} & \textbf{\#M} & \textbf{\#L} & \textbf{\#U} & \textbf{\#G} & \textbf{\#A} \\ \hline
        \multirow{4}{=}{Stack-Overflow} & avg & 12 & 1.88 & 0.42 & 0.79 & 0.03 & 0.27 & 0.36 & 0 \\
        & med & 10 & 1 & 0 & 1 & 0 & 0 & 0 & 0 \\
        & min & 4 & 1 & 0 & 0 & 0 & 0 & 0 & 0 \\
        & max & 33 & 5 & 1 & 2 & 1 & 2 & 2 & 0 \\ \hline
        \multirow{4}{=}{Official Document} & avg & 8 & 1.15 & 0.31 & 0.42 & 0.04 & 0.08 & 0.27 & 0.04 \\
        & med & 7 & 1 & 0 & 0 & 0 & 0 & 0 & 0 \\
        & min & 4 & 1 & 0 & 0 & 0 & 0 & 0 & 0 \\
        & max & 17 & 3 & 1 & 1 & 1 & 2 & 1 & 1 \\ \hline
        \multirow{4}{=}{Twitter API} & avg & 16 & 2.6 & 0.8 & 0 & 0 & 1 & 0.6 & 0.2 \\
        & med & 14 & 2 & 1 & 0 & 0 & 1 & 1 & 0 \\
        & min & 9 & 2 & 0 & 0 & 0 & 0 & 0 & 0 \\
        & max & 26 & 4 & 1 & 0 & 0 & 2 & 1 & 1 \\ \hline
        \multirow{4}{=}{Kaggle} & avg & 13 & 3.2 & 0.7 & 0.52 & 0 & 1.54 & 0.43 & 0 \\
        & med & 12 & 3 & 1 & 0 & 0 & 2 & 0 & 0 \\
        & min & 8 & 2 & 0 & 0 & 0 & 0 & 0 & 0 \\
        & max & 27 & 6 & 1 & 2 & 0 & 3 & 2 & 0 \\ \hline
        \multirow{4}{=}{\textbf{Total}} & avg & 12 & 2.29 & 0.53 & 0.55 & 0.02 & 0.79 & 0.38 & 0.02 \\
        & med & 11 & 2 & 1 & 1 & 0 & 1 & 0 & 0 \\
        & min & 4 & 1 & 0 & 0 & 0 & 0 & 0 & 0 \\
        & max & 33 & 6 & 1 & 2 & 1 & 3 & 2 & 1 \\ \hline

    \end{tabular}
\vspace{-10pt}
\end{table}

We have collected 110 benchmarks from 4 representative sources, i.e., StackOverflow, MongoDB official document, Twitter API documents, and Kaggle competitions, which cover a wide spectrum of realistic scenarios.


{
\begin{itemize}[leftmargin=*]
\item StackOverflow.
The StackOverflow dataset is adapted from StackOverflow posts where developers ask about real-world problems. Each post in our dataset has 453K visits, 4 answers, and 127 votes on average, which demonstrates these queries attract lots of attention from the community. Most of the the examples and constants are extracted from the post content. If some post does not provide enough examples, we add the examples.

\item MongoDB Document.
The MongoDB official documents cover a representative set of queries that the MongoDB community believes are commonly used in practice. The examples and constants are all collected from the example section of official documents.

\item Twitter API.
The Twitter dataset consists of tweets and user replies which mainly focus on calculating tweet statistics, such as the count of replies. The benchmarks represent typical scenarios for data analysts to get information from social networks and online forums. The examples are collected from the response data of APIs.

\item Kaggle.
The Kaggle dataset contains information about satellite images, where benchmarks reflect scenarios for scientific research, such as extracting different labels for training machine learning models and collecting statistics. The examples are sampled from the provided JSON file.

\end{itemize}
}

{
Table~\ref{table:benchmark} summarizes the statistics of these datasets. Among these datasets, Twitter API and Kaggle benchmarks are more complex than StackOverflow and MongoDB Document in terms of the number of attributes, collection sizes, etc.


To further understand the complexity of benchmarks, we have also collected the statistics on the ground truth queries in Table \ref{table:ground-truth-query}.
The maximum AST size of ground truth queries is 33 among all benchmarks, and the average is 12. Over half of the ground truth queries have an AST size larger than 10. This indicates a high level of complexity, as longer queries typically require synthesizers to explore a larger search space. 
Furthermore, the number of operators (or pipeline stages) in a single query ranges from 1 to 6. Frequently occurring operators include \cmdProject, \cmdMatch, \cmdUnwind, and \cmdGroup. Notably, \cmdUnwind and \cmdGroup pose significant challenges for synthesis, as they can substantially change the structure of collections and documents.
In contrast, the \cmdLookup operator appears infrequently in ground truth queries. This is consistent with the typical usage of document databases where users try to avoid ``join'' operations between multiple collections. Similarly, the \cmdAddFields operator is also not used frequently in our datasets. 

}


\subsection{Effectiveness and Efficiency}

\begin{table}[!t]
\caption{Evaluation results for \tool. \textbf{\#n} and \Checkmark denote the number of benchmarks and solved benchmarks, and \textbf{time} indicates the time (in seconds) to solve benchmarks. \textbf{\#$\sketch$}, \textbf{\#$\prog$}, \textbf{\#size} refer to the number of sketches, complete programs, and AST nodes of synthesized programs, respectively.}
\label{table:main-eval}
\centering
    \begin{tabular}{|p{4em}|c|c|c|c|c|c|c|}
        \hline
        \textbf{dataset} & \textbf{\#n} & \Checkmark & & \textbf{time (s)} & \textbf{\#$\sketch$} & \textbf{\#$\prog$} & \textbf{\#size}\\ \hline
        \multirow{4}{=}{Stack-Overflow} & \multirow{4}{*}{33} & \multirow{4}{*}{33} & avg & 9.2 & 86 & 31 & 12 \\
        & & & med & 2.6 & 15 & 6 & 11 \\
        & & & min & 0.5 & 2 & 1 & 4 \\
        & & & max & 184.5 & 854 & 308 & 32 \\ \hline
        \multirow{4}{=}{MongoDB Document} & \multirow{4}{*}{26} & \multirow{4}{*}{26} & avg & 5.7 & 11 & 53 & 9 \\
        & & & med & 1.1 & 6 & 14 & 10  \\
        & & & min & 0.5 & 2 & 1 & 4  \\
        & & & max & 78.6 & 124 & 576 & 19  \\ \hline
        \multirow{4}{=}{Twitter API} & \multirow{4}{*}{5} & \multirow{4}{*}{5} & avg & 10.5 & 81 & 61 & 15  \\
        & & & med & 10.1 & 36 & 81 & 15  \\
        & & & min & 1.8 & 15 & 1 & 9  \\
        & & & max & 19.9 & 165 & 131 & 22  \\ \hline
        \multirow{4}{=}{Kaggle} & \multirow{4}{*}{46} & \multirow{4}{*}{44} & avg & 23.4 & 350 & 79 & 16  \\
        & & & med & 6.8 & 160 & 4 & 13  \\
        & & & min & 1.0 & 8 & 1 & 8  \\
        & & & max & 201.6 & 3975 & 1235 & 38  \\ \hline
        \multirow{4}{=}{\textbf{Total}} & \multirow{4}{*}{110} & \multirow{4}{*}{108} & avg & 14.2 & 175 & 57 & 13  \\
        & & & med & 3.2 & 31 & 7 & 11  \\
        & & & min & 0.5 & 2 & 1 & 4  \\
        & & & max & 201.6 & 3975 & 1235 & 38  \\ \hline
    \end{tabular}
\vspace{-10pt}
\end{table}




The evaluation results and the statistics of synthesized programs are presented in Table~\ref{table:main-eval}. Given a time limit of 5 minutes, \tool can solve 108 out of 110 benchmarks and only gets timeout on two challenging benchmarks (both in Kaggle). Note that the ground-truths of these two benchmarks are more complex than the others from our manual inspection. This serves as evidence of the effectiveness of \tool in synthesizing document database queries from examples. Further, \tool can solve most benchmarks in an average of {14.2} seconds as shown in Table~\ref{table:main-eval}.
Furthermore, observing the number of sketches \#$\sketch$ and complete programs \#$\prog$, \tool iterates over 175 sketches but only completes {57} full programs on average. It demonstrates that our synthesis technique based on collection abstractions is efficient in pruning infeasible sketches and thus speeds up the synthesis process.

\bfpara{Qualitative analysis.} 
We observe that the number of attributes in the document, the depth of the document, the number of constants, and the query complexity affect the synthesis time.
For instance, the Kaggle dataset needs longer synthesis time than others because the benchmark has a large number of attributes and the documents are deeply nested.
In general, more complex queries need the synthesizer to iterate more sketches. More attributes, deeper nesting, and more constants require enumerating more queries while completing the sketch.


\bfpara{Non-desired programs.} 
To understand if \tool can synthesize desired queries, we have manually inspected all 108 synthesized queries and found 107 of them are equivalent to the desired ones. There is only one benchmark (from StackOverflow) where \tool synthesized a plausible query in terms of the example but the query is not desired.
The reason is that this benchmark involves a complex predicate that requires numerous unseen examples to eliminate mismatch cases. However, only a few examples are provided on the StackOverflow post, so \tool cannot find the desired predicate but synthesize an alternative satisfying the examples.

\textbox{\textbf{Answer to RQ1}: \tool successfully synthesizes 108 out of 110 benchmarks from examples and the average synthesis time is {14.2} seconds.}

\subsection{Ablation Study}

\begin{figure}[!t]
    \centering
    \includegraphics[width=\linewidth]{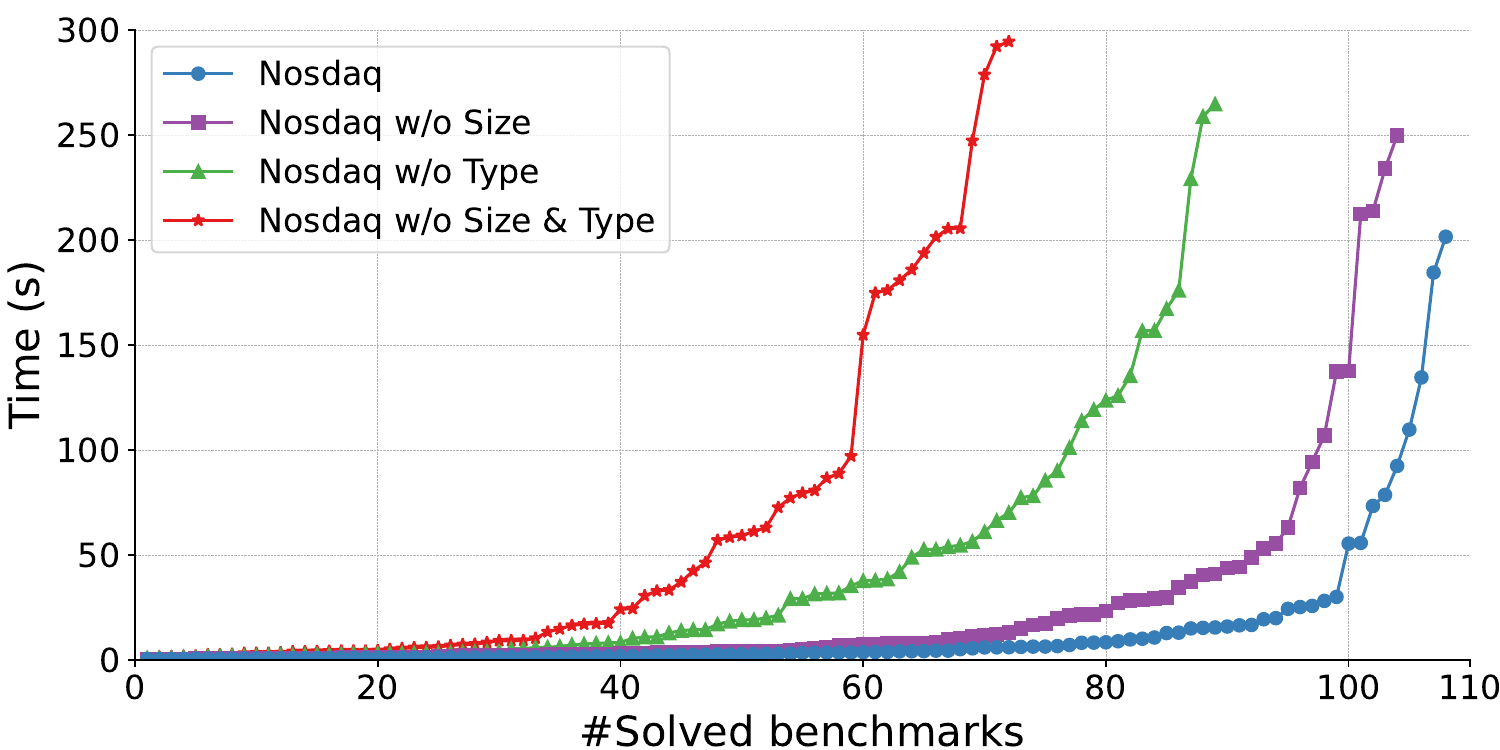}
    \caption{Ablation study.}
    \label{fig:ablation-study}
    \vspace{-10pt}
\end{figure}



To understand how the type and size information in collection abstractions may affect the efficiency, we perform an ablation study. Specifically, we have created three variants of \tool that disable (1) the size information, (2) the type information, and (3) both size and type information in the abstraction. We run all these variants on the 110 benchmarks and obtain the result shown in Figure~\ref{fig:ablation-study}, where a point $(x, y)$ means a variant can synthesize $x$ benchmarks and the time for each benchmark is within $y$ seconds.
As shown in the figure, without size in the abstraction, the variant times out on {4} more benchmarks and requires approximately {10} seconds longer on average. Without type, the variant triggers timeout on {19} more benchmarks and requires around {27} seconds longer on average to complete the synthesis process.
This indicates that the document type in the collection abstraction significantly improves the synthesis time.


\textbox{\textbf{Answer to RQ2}: Both type and size information can make \tool more efficient but the former is more significant.}

\subsection{Comparison with Baselines}


\begin{figure}[!t]
    \centering
    \includegraphics[width=\linewidth]{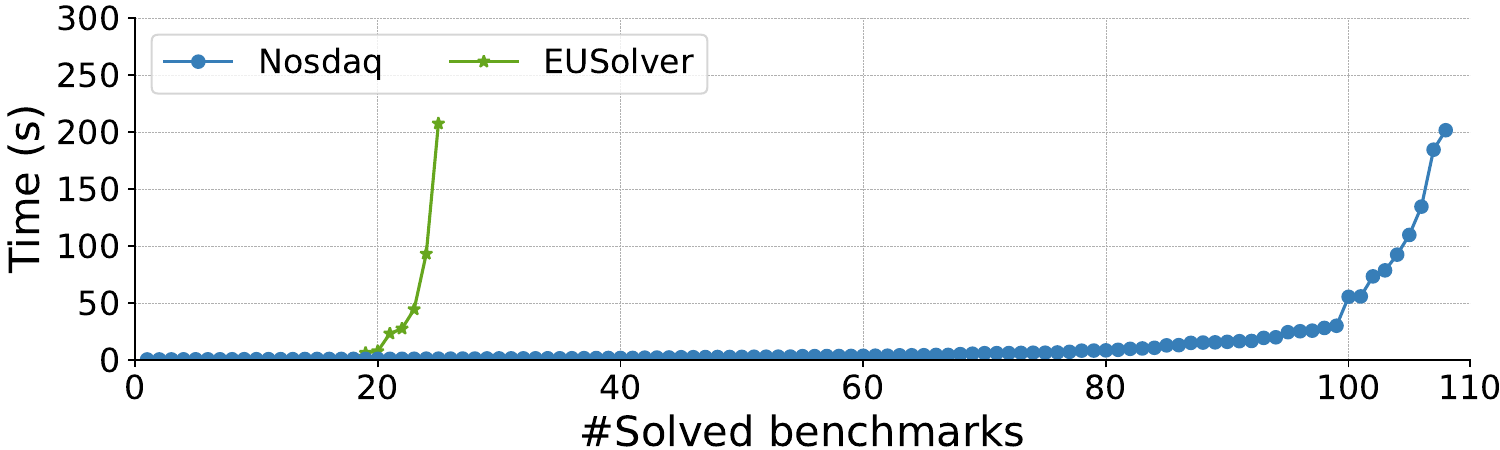}
    \caption{Comparison between \tool and \eusolver.}
    \label{fig:baseline}
    \vspace{-15pt}
\end{figure}


{
To compare \tool with a baseline, we have instantiated the \eusolver framework~\cite{eusolver-tacas17} to synthesize document database queries from examples. 
As a generic solver, \eusolver can be easily extended to support documents and collections in the specification, since it provides necessary support for lists and maps. Secondly, \eusolver remains a competitive baseline in program synthesis, as evidenced from recent work~\cite{BarnabyCSD23, NazariHSRR23, JiKXH23}.
As shown in Figure~\ref{fig:baseline}, as opposed to 108 benchmarks solved by \tool, \eusolver can only solve {25} benchmarks within the 5-minute time limit due to the large search space of document database queries in general.

}



{
To compare \tool with the LLM-based approach, we have used ChatGPT (version gpt-4o-2024-08-06) to synthesize all of our 110 benchmarks. Specifically, we have used the same set of input-output examples and constants in each benchmark and asked ChatGPT to generate MongoDB queries. To make fair comparisons, we did not provide additional natural language descriptions about what the query should do. The evaluation shows that GPT can only generate the desired query for 53 out of 110 benchmarks. For 24 benchmarks, the generated query is plausible but undesired, i.e., it is consistent with the examples but not equivalent to the desired one. For the remaining 33 benchmarks, the generated query is inconsistent with the input-output examples. The errors made by GPT include misunderstanding the semantics of operators, missing predicates, etc. Recall that \tool can synthesize desired queries for 107 benchmarks and plausible but undesired query for 1 benchmark. We believe our synthesis technique is more effective and generalizable than GPT to synthesize document database queries from examples.
}

\textbox{\textbf{Answer to RQ3}: \tool can solve 108 out of 110 benchmarks, whereas \eusolver can only solve {25} benchmarks, {and ChatGPT-4o can solve 77 benchmarks.}}

\subsection{Impact of Collection Size}
\begin{figure}[!t]
    \centering
    \includegraphics[width=\linewidth]{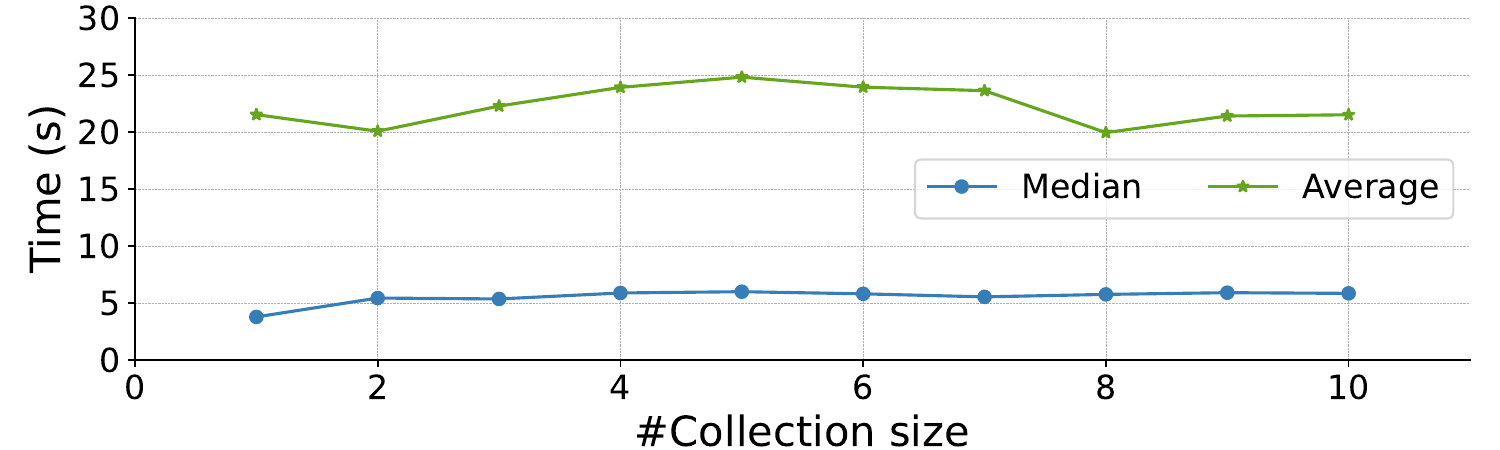}
    \caption{Impact of collection size on synthesis time.}
    \label{fig:collection-size-impact-time}
    \vspace{-10pt}
\end{figure}

\begin{figure}[!t]
    \centering
    \includegraphics[width=\linewidth]{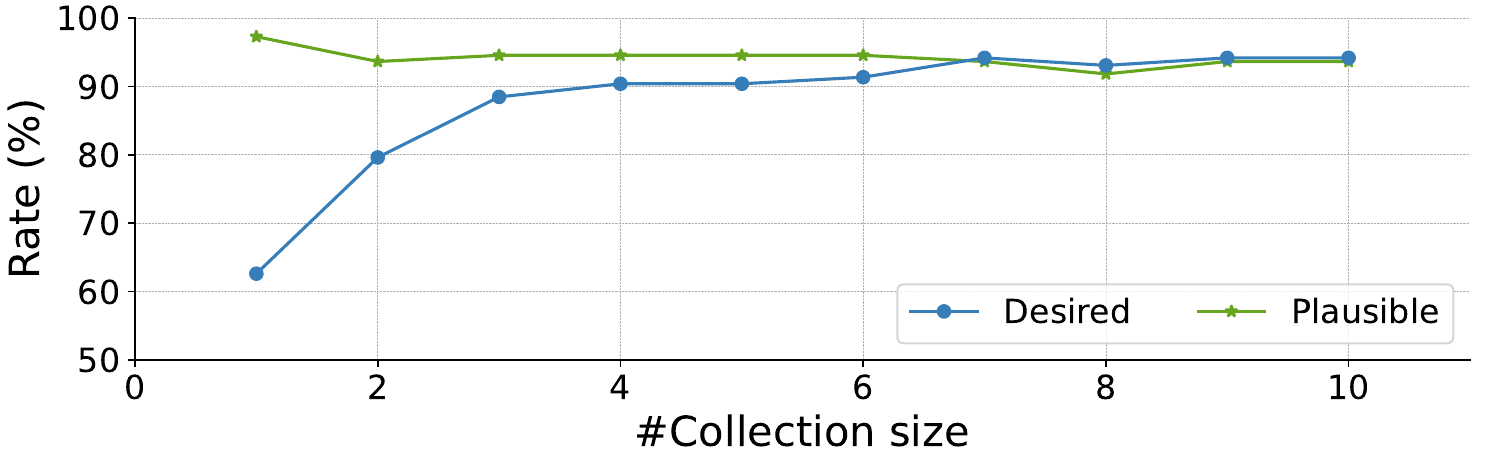}
    \caption{Impact of collection size on rates of plausible and desired queries.}
    \label{fig:collection-size-impact-rate}
    \vspace{-10pt}
\end{figure}

To analyze the impact of collection size on the performance of \tool, we have conducted experiments across all 110 benchmarks to evaluate how different collection sizes influence \tool's behavior. Specifically, we sampled 10 documents for each collection and ran \tool on variants with collection sizes ranging from 1 to 10 documents. The impact on synthesis time is presented in Figure~\ref{fig:collection-size-impact-time}, while the impact on rates of plausible and desired queries are shown in Figure~\ref{fig:collection-size-impact-rate}.

The plausible rate is defined as the ratio of benchmarks synthesized within a 5-minute time limit to the total number of benchmarks. The desired rate represents the ratio of benchmarks for which the synthesized query is equivalent to the desired one to the total number of synthesized benchmarks.

As shown in the figures, the synthesis time of \tool remains relatively insensitive to changes in collection size within the range of 1 to 10 documents in each collection. Similarly, the plausible rate also remains stable. In contrast, the desired rate shows a significant increase when the collection size grows from 1 to 3, after which it stabilizes. This can be attributed to the fact that smaller collection sizes provide insufficient examples to synthesize the desired query, leading to simpler queries that are plausible but not desired.

\textbox{\textbf{Answer to RQ4}: The synthesis time of \tool demonstrates minimal sensitivity to changes in collection size. The rate of synthesizing a desired query increases rapidly as the collection size grows from 1 to 3 and stabilizes thereafter.}

\subsection{Threats to Validity}



First, although we believe our datasets are representative, which are obtained from various real-world scenarios, our evaluation results are limited to the collected datasets. The \tool tool might perform differently on other datasets.
Second, our domain-specific language only corresponds to a core subset of the MongoDB aggregation pipeline. While it is convenient to extend the abstract semantics to other query operators, the performance of the tool might be different due to the change in SMT formulas for symbolic reasoning.
Third, all the experiments are conducted on a machine as specified in Section~\ref{sec:eval}. Running the experiments on a different machine may yield different results.

\section{Related Work} \label{sec:related}

\bfpara{Program synthesis for software engineering.}
Program synthesis techniques have been applied to address various software engineering problems, such as program refactoring~\cite{Raychev-oopsla13,soar-icse21,solidare-oopsla22,revamp-popl24}, program repair~\cite{spr-fse15,angelix-icse16,Xiong-icse17,pyter-fse22}, code completion~\cite{codehint-icse14,slang-pldi14}, software testing~\cite{syrust-pldi21,javatailor-icse22}, and so on.
This paper focuses on the topic of generating document database queries from input-output examples.

\vspace{-2pt}
\bfpara{Synthesizing database queries.}
Among related papers, the most related is a body of work on synthesizing database queries. \sqlsynthesizer~\cite{SQLSynthesizer-ase2013}, \scythe~\cite{scythe-pldi2017} and \patsql~\cite{patsql-vldb2021} synthesize SQL queries for relational databases from examples, while \sickle~\cite{sickle-pldi2022} synthesizes analytical SQL queries given computation demonstrations. \sqlizer~\cite{sqlizer-oopsla2017} considers nature language description as the specification for SQL query synthesis. However, none of the prior work can synthesize queries of document databases such as MongoDB.

\vspace{-2pt}
\bfpara{Synthesis with deduction.}
A line of work performs deduction to prune infeasible programs in program synthesis~\cite{morpheus-pldi17, neo-pldi2018, concord-cav2020, DryadStnth-pldi2020, ngds-iclr2018, regel-pldi2020, lambad2-pldi2015, FlashExtract-pldi2014, myth-pldi2015, FlashMeta-oospla2015}. For example, \morpheus~\cite{morpheus-pldi17} and \neo~\cite{neo-pldi2018} utilize SMT-based deduction that generates formulas based on semantics and input-output examples to prune infeasible programs. \ngds~\cite{ngds-iclr2018} and \concord~\cite{concord-cav2020} combine deduction and machine learning techniques to prune the search space. 
{
\tool adapts the high-level approach of \morpheus~\cite{morpheus-pldi17} and \neo~\cite{neo-pldi2018} to the setting of document database queries. However, \morpheus mainly focuses on tabular data, whereas \tool focuses on hierarchical data. The abstraction used by \morpheus is related to the number of rows and columns of tables. This abstraction cannot be directly used for deduction in a synthesizer that aims to generate document database queries, because these queries operate over more involved hierarchical data. Therefore, \tool uses the novel abstraction consisting of hierarchically nested types for its documents and the collection size, which is one of the main contributions of this paper.
}

\vspace{-2pt}
\bfpara{Synthesis with abstraction.}
Another line of related work is to synthesize programs using abstractions~\cite{blaze-popl2017, simpl-sas17, tygar-popl2019, Mell-popl24, Vechev-popl10}.
For example, \simpl~\cite{simpl-sas17} uses abstract interpretation to guide the synthesis of imperative programs from examples. Mell et al.~\cite{Mell-popl24} also use abstract interpretation for optimal program synthesis.
\blaze~\cite{blaze-popl2017} constructs and iteratively refines the abstract finite tree automata that represent a set of programs. This approach iteratively prunes and refines automata when the corresponding programs do not satisfy examples, until a correct program is found. Unlike prior techniques, \tool employs a novel collection abstraction to represent complex hierarchical data (e.g., BSON) and uses abstract semantics to rule out infeasible sketches representing a large set of programs.

\vspace{-2pt}
\bfpara{Wrangling semi-structured data.}
Various techniques have been proposed to wrangle semi-structured data, such as JSON, XML documents, spreadsheets, and log files. For example, there is a line of work~\cite{xmldb-sigmod2001, xmldb-sigmod2002, xmldb-widm2004, xmldb-is2007} that aims to map XML documents to relational data for query processing. \datamaran~\cite{datamaran-sogmod2018} converts the semi-structured log into a structured relational format. \flashExtract~\cite{FlashExtract-pldi2014} extracts relevant data from text files, websites, and spreadsheets.
\flashRelate~\cite{FlashRelate-plid2015} extracts relational data from semi-structured spreadsheets by examples.
\treex~\cite{treex-wsdm17} synthesizes extractors for real-world large-scale websites.
Since document databases store semi-structured data by nature, \tool can also be viewed as a query synthesizer over semi-structured data. However, different from prior work, \tool focuses on core query language of document databases and aims to address significant challenges raised by specialized operators such as \cmdGroup, \cmdUnwind, and \cmdLookup.  


\vspace{-2pt}
\bfpara{Synthesizing data transformation scripts.}
Many synthesizers aim to automatically generate data transformation scripts from high-level specifications~\cite{hades-pldi2016,mitra-vldb2018,morpheus-pldi17,autopandas-oopsla2019,Flashfill-popl2011,wrex-chi2020,ProgFromEx-pldi2011,Trinity-vldb2019}. For example, \hades~\cite{hades-pldi2016} synthesizes scripts to handle hierarchically structured data such as file systems, XML, and HDF files.
\mitra~\cite{mitra-vldb2018} aims to synthesize scripts to convert hierarchical data into relational tables.
\dynamite~\cite{dynamite-vldb2020} transforms data between various types of databases by synthesizing Datalog programs.
In contrast, \tool is designed to handle complex data structures in document databases and leverage collection abstractions to efficiently synthesize queries from examples, which is beyond the capability of prior work. For instance, \hades focuses on structure changes in the transformation but does not support aggregations, but \tool can synthesize aggregate queries with \cmdGroup operations.


\section{Conclusion} \label{sec:concl}

This paper presents a technique that automatically synthesizes document database queries from input-output examples. To achieve better performance, we develop a novel abstraction for collections containing hierarchical and nested data structures and leverage this abstraction for deduction to prune the search space of target queries. An evaluation of 110 benchmarks from various sources demonstrates our technique is effective and efficient in solving 108 benchmarks.

\balance
\bibliographystyle{plain}
\bibliography{main}

\clearpage
\appendices
\section{Auxiliary Functions and Definitions}

\vspace{-15pt}
\begin{figure}[h]
\scriptsize
\footnotesize
\[
\begin{array}{c}

\irulelabel
{\begin{array}{c}
\tau_a = \set{a_{a1}: \tau_{a1}, \ldots, a_{an}: \tau_{an}, a_1: \tau_1, \ldots, a_p: \tau_p}\\
\tau_b = \set{a_{b1}: \tau_{b1}, \ldots, a_{bm}: \tau_{bm}, a_1: \tau'_1, \ldots, a_p: \tau'_p}\\
\end{array}}
{\begin{array}{c}
\tau_a \cup \tau_b = \set{a_1: \tau_{a1}, \ldots, a_{an}: \tau_{an}, b_1: \tau_{b1}, \ldots, a_{bm}: \tau_{bm}\\
 \qquad a_1: \tau_1 \cup \tau'_1, \ldots, a_p: \tau_p \cup \tau'_p}\\
\end{array}}
{\textrm{(T-Union)}}

\\\\\

\irulelabel
{\begin{array}{c}
\tau_a \notin \augType \lor 
\tau_b \notin \augType\\
\tau_a \ne \tau_b
\end{array}}
{\begin{array}{c}
\tau_a \cup \tau_b = \bot\\
\end{array}}
{\textrm{(T-Union-Ne)}}


\irulelabel
{\begin{array}{c}
\tau_a \notin \augType \lor 
\tau_b \notin \augType\\
\tau_a = \tau_b
\end{array}}
{\begin{array}{c}
\tau_a \cup \tau_b = \tau_a\\
\end{array}}
{\textrm{(T-Union-Eq)}}

\\\\\

\irulelabel
{\begin{array}{c}
\tau_a = \set{a_{a1}: \tau_{a1}, \ldots, a_n: \tau_{an}, a_1: \tau_1, \ldots, a_p: \tau_p}\\
\tau_b = \set{a_{b1}: \tau_{b1}, \ldots, a_m: \tau_{bm}, a_1: \tau'_1, \ldots, a_p: \tau'_p}\\
\end{array}}
{\begin{array}{c}
\tau_a \cap \tau_b = \set{a_1: \tau_1 \cap \tau'_1, \ldots, a_p: \tau_p \cap \tau'_p}\\
\end{array}}
{\textrm{(T-Inter)}}

\\\\\

\irulelabel
{\begin{array}{c}
\tau_a \notin \augType \lor 
\tau_b \notin \augType\\
\tau_a \ne \tau_b
\end{array}}
{\begin{array}{c}
\tau_a \cap \tau_b = \bot\\
\end{array}}
{\textrm{(T-Inter-Ne)}}


\irulelabel
{\begin{array}{c}
\tau_a \notin \augType \lor 
\tau_b \notin \augType\\
\tau_a = \tau_b
\end{array}}
{\begin{array}{c}
\tau_a \cap \tau_b = \tau_a\\
\end{array}}
{\textrm{(T-Inter-Eq)}}

\\\\\

\irulelabel
{\begin{array}{c}
\tau = \set{a_1: \tau_1, \ldots, a_n: \tau_n, a: \bot}
\end{array}}
{\tau = \set{a_1: \tau_1, \ldots, a_n: \tau_n}}
{\textrm{(T-Bot)}}

\\\\\

\irulelabel
{}
{\bot \cup \type = \type}
{\textrm{(T-Bot-U)}}

\irulelabel
{\begin{array}{c}
\tau_a \cap \tau_b = \tau_a\\
\tau_a \cup \tau_b = \tau_b
\end{array}}
{\tau_a \subseteq \tau_b}
{\textrm{(T-Subset)}}

\\\\\

\irulelabel
{\begin{array}{c}
\tau = \set{a_1: \tau_1, \ldots, a_n: \tau_n}\\
(\exists i. a_i = a) \lor (\exists i. a \in \tau_i \land \tau_i \in \docType)\\
i = 1, \ldots, n
\end{array}}
{a \in \tau}
{\textrm{(T-In)}}

\\\\\

\irulelabel
{\begin{array}{c}
\tau \cup \tau_b = \tau_a \cup \tau_b\\
\tau \cap \tau_b = \set{} \qquad
\tau \subseteq \tau_a\\
\end{array}}
{\tau_a - \tau_b = \tau}
{\textrm{(T-Sub)}}

\irulelabel
{\begin{array}{c}
a \notin \tau
\end{array}}
{\begin{array}{c}
\tau[\tau'/a] = \tau
\end{array}
}
{\textrm{(T-Rep-Not-In)}}

\\\\\

\irulelabel
{\begin{array}{c}
\tau = \set{a_1: \tau_1, \ldots, a_n: \tau_n}\\
\exists i. a_i = a \\
i = 1, \ldots, n
\end{array}}
{\begin{array}{c}
\tau[\tau'/a] = \set{a_1: \tau_1, \ldots, a_{i-1}: \tau_{i-1}, \\
\qquad \qquad  a_{i+1}: \tau_{i+1}, \ldots, a_n: \tau_n, a_i: \tau'}\\
\end{array}
}
{\textrm{(T-Rep-In-1)}}

\\\\\

\irulelabel
{\begin{array}{c}
\tau = \set{a_1: \tau_1, \ldots, a_n: \tau_n}\\
\exists i. a \in \tau_i \\
i = 1, \ldots, n
\end{array}}
{\begin{array}{c}
\tau[\tau'/a] = \set{a_1: \tau_1, \ldots, a_{i-1}: \tau_{i-1}, \\
\qquad \qquad  a_{i+1}: \tau_{i+1}, \ldots, a_n: \tau_n, a_i: \tau_i[\tau'/a]}\\
\end{array}
}
{\textrm{(T-Rep-In-2)}}

\end{array}
\]
\vspace{-10pt}
\caption{Type Operations}
\label{fig:type-operations}
\vspace{-10pt}
\end{figure}

\begin{definition} \label{def:formal-match}
Let $\type$ be a document type and $\augType$ be an augmented type where $\type = \set{a_1: \type_1, \ldots, a_n: \type_n}$ and $\augType = \set{b_1: T_1, \ldots, b_m: T_m, \texttt{?}_{m+1}^1: T_{m+1}, \ldots, \texttt{?}_{m+p}^1: T_{m+p}, \texttt{?}_{m+p+1}^+: T_{m+p+1}, \ldots, \texttt{?}_{m+p+q}^+: T_{m+p+q}}$. We say $\type$ matches $\augType$, denoted $\type \matches \augType$, if and only if 
\begin{enumerate}
    \item There exists a map $M_m$ from the top-level attribute in $\augType$ to the top-level attribute in $\type$ such that $\forall i. 1 \le i \le m \Rightarrow M_m[b_i] = a \land b_i = a \land \textsl{Type}(b_i) = \textsl{Type}(a)$
    \item There exists a map $M_p$ from the top-level attribute in $\augType$ to the top-level attribute in $\type$ such that $\forall i. m+1 \le i \le m+p+1 \Rightarrow M_p[b_i] = a \land (\textsl{Type}(b_i) = \textsl{Type}(a) \lor \textsl{Type}(b_i) = \cmdAny)$
    \item There exists a map $M_q$ from the top-level attribute in $\augType$ to a set of top-level attributes in $\type$ such that $\forall i. m+p+1 \le i \le m+p+q \Rightarrow M_q[b_i] = S_a \land S_a \ne \set{} \land  (\forall a. a \in S_a \Rightarrow (\textsl{Type}(b_i) = \textsl{Type}(a) \lor \textsl{Type}(b_i) = \cmdAny))$.
    \item All the maps are one-to-one correspondence. Formally, $\forall M. M \in \set{M_m, M_p, M_q} \Rightarrow (\forall d. \forall e. d \in M \land e \in M \land d \ne e \Rightarrow M[d] \ne M[e])$
    \item $K_m \cap K_p = \set{} \land K_p \cap K_q = \set{} \land K_m \cap K_q = \set{} \land K_m \cup K_p \cup K_q = A_\augType$ where $K_m, K_p, K_q$ are the key set of $M_m, M_p, M_q$ and $A_\augType = \set{b_1, \ldots, b_m, \texttt{?}_{m+1}^1, \ldots, \texttt{?}_{m+p}^1,\\ \texttt{?}_{m+p+1}^+, \ldots, \texttt{?}_{m+p+q}^+}$
    \item $V_m \cap V_p = \set{} \land V_p \cap (\bigcup V_q) = \set{} \land V_m \cap (\bigcup V_q) = \set{} \land V_m \cup V_p \cup (\bigcup V_q) = A_\type$ where $V_m, V_p, V_q$ are the value set of $M_m, M_p, M_q$ and $A_\type = \set{a_1, \ldots, a_n}$.
\end{enumerate}
\end{definition}

\section{Concrete Formal Semantics}

\begin{figure}[!t]
\scriptsize

\framebox[\columnwidth]{$\denot{\query} :: \textsl{Database } \db \rightarrow \textsl{Collection } \collection$}
\[
\begin{array}{rcl}
\denot{N}_{\db} &=& \db[N] \\
\denot{\cmdProject(\prog, \vec{h})}_{\db} & = &  \cmdMap(\lambda \doc. \cmdExtractAttrs(\doc, \vec{h}), \denot{\prog}_{\db})\\

\denot{\cmdMatch(\prog, \pred)}_{\db} & = &  \cmdFilter(\lambda \doc. \denot{\pred}_{\doc} = \top, \denot{\prog}_{\db})\\

\denot{\cmdAddFields(\prog, \vec{h}, \vec{E})}_{\db} & = &  \cmdMap(\lambda \doc. \cmdAddAttrs(\doc, \vec{h}, \vec{\denot{E}_\doc}), \denot{\prog}_{\db})\\

\denot{\cmdUnwind(\prog, h)}_{\db} & = &  \cmdFlatMap(\lambda \doc. \cmdFlattenArr(\doc, h), \denot{\prog}_{\db})\\

\denot{\cmdGroup(\prog, \vec{k}, \vec{a}, \Vec{A})}_{\db} & = & \cmdMap(\lambda g.\cmdAddAttrs(\set{\_\texttt{id}: g}, \vec{[a]}, \\
        & & \quad \denot{\vec{A}}_{\cmdFilter(\lambda D. \cmdExtractAttrs(\doc, \vec{k})=g, \denot{\prog}_{\db})}), \\
        & & \quad \cmdDedup(\cmdMap(\lambda \doc.\cmdExtractAttrs(\doc, \vec{k}),   \denot{\prog}_{\db})))\\

\denot{\cmdLookup(\prog, k, h, N, a)}_{\db} & = &  \cmdMap(\lambda \doc. \doc[a \mapsto C], \denot{\prog}_{\db})~\text{where}\\
&&  C = \cmdFilter(\lambda F. \denot{k}_{\doc} = \denot{h}_F, \db[N])\\

\end{array}
\]

\framebox[\columnwidth]{$\denot{\pred} :: \textsl{Document } \doc \rightarrow \cmdBool$}
\[
\begin{array}{rcl}
\denot{b}_{\doc} &=& b~\text{where}~b \in \set{\top, \bot}\\
\denot{h \circ c}_{\doc} &=& \denot{h}_{\doc} \circ c~\text{where}~\circ \in \set{=, \neq}\\
\denot{h < c}_{\doc} &=& \cmdIte(\denot{h}_{\doc} = \cmdNull \lor c = \cmdNull, \bot,  \denot{h}_{\doc} < c)\\
\denot{h > c}_{\doc} &=& \cmdIte(\denot{h}_{\doc} = \cmdNull \lor c = \cmdNull, \bot,  \denot{h}_{\doc} > c)\\
\denot{h \leq c}_{\doc} &=& \denot{h < c}_{\doc} \lor \denot{h = c}_{\doc}\\
\denot{h \geq c}_{\doc} &=& \denot{h > c}_{\doc} \lor \denot{h = c}_{\doc}\\
\denot{\cmdSize(h, c)}_{\doc} &=& \cmdIte(\denot{h}_\doc = \cmdNull \lor c = \cmdNull, \bot, |\denot{h}_{\doc}| = c) \\
\denot{\cmdExists(h)}_{\doc} &=& \cmdHasAp(D, h) \\
\denot{\pred_1 \land \pred_2}_{\doc} &=& \denot{\pred_1}_{\doc} \land \denot{\pred_2}_{\doc} \\
\denot{\pred_1 \lor \pred_2}_{\doc} &=& \denot{\pred_1}_{\doc} \lor \denot{\pred_2}_{\doc} \\
\denot{\neg \pred}_{\doc} &=& \neg \denot{\pred}_{\doc} \\
\end{array}
\]

\framebox[\columnwidth]{$\denot{E} :: \textsl{Document } \doc \rightarrow \cmdValue$}
\[
\begin{array}{rcl}
\denot{h}_\doc &=& \cmdGet(D, h) \\
\denot{h_1 \oplus h_2}_{\doc} &=& \cmdIte(\denot{h_1}_\doc =\cmdNull \lor \denot{h_2}_\doc = \cmdNull, \cmdNull, \denot{h_1}_\doc \oplus \denot{h_2}_\doc)\\
\denot{f(h)}_\doc &=& f(\denot{h}_\doc) \\
\end{array}
\]

\framebox[\columnwidth]{$\denot{A} :: \textsl{Document List } \vec{\doc} \rightarrow \cmdValue$}
\[
\begin{array}{rcl}
\denot{\cmdSum(h)}_{\vec{\doc}} &=& \text{let}~ xs = \cmdMap(\lambda \doc. \denot{h}_\doc, \vec{\doc}) ~\text{in}~ \\
& & \hspace{.4em} \sum_{x \in xs} \cmdIte(x = \cmdNull, 0, x)\\
\denot{\cmdMin(h)}_{\vec{\doc}} &=& \text{let}~ xs = \cmdMap(\lambda \doc. \denot{h}_\doc, \vec{\doc}) ~\text{in}~ \\
& & \hspace{.4em} \cmdIte(\cmdAllNull(xs), \cmdNull, \tmin_{x \in xs} \cmdIte(x = \cmdNull, +\infty, x)) \\
\denot{\cmdMax(h)}_{\vec{\doc}} &=& \text{let}~ xs = \cmdMap(\lambda \doc. \denot{h}_\doc, \vec{\doc}) ~\text{in}~ \\
& & \hspace{.4em} \cmdIte(\cmdAllNull(xs), \cmdNull, \tmax_{x \in xs} \cmdIte(x = \cmdNull, -\infty, x)) \\
\denot{\cmdAvg(h)}_{\vec{\doc}} &=& \text{let}~ xs = \cmdMap(\lambda \doc. \denot{h}_\doc, \vec{\doc}) ~\text{in}~ \\
& & \hspace{.4em} \cmdIte(\cmdAllNull(xs), \cmdNull, \denot{\cmdSum(h)}_{\vec{\doc}} \\
& & \hspace{.4em} / \sum_{x \in xs} \cmdIte(x = \cmdNull, 0, 1)) \\
\denot{\cmdCount()}_{\vec{\doc}} &=& |\vec{\doc}| \\
\end{array}
\]

\vspace{-5pt}
\caption{Semantics of database queries. $\cmdExtractAttrs(D, \vec{h})$ constructs a document by extracting access paths $\vec{h}$ and their values from document $D$. $\cmdAddAttrs(D, \vec{h}, \vec{v})$ returns a document by adding new access paths $\vec{h}$ with values $\vec{E}$ to document $D$. $\cmdFlattenArr(D, h)$ requires $h$ to be an array in document $D$. It maps $D$ to a list of documents where the value of $h$ in the $i$-th document is the $i$-th element in the original array. $\cmdHasAp(D, h)$ checks whether there exists an access path $h$ in document $D$. 
$\cmdGet(D, h)$ gets the value of access path $h$ from document $D$.}
\label{fig:concrete-semantics}
\vspace{-10pt}
\end{figure}

The denotational semantics of our query language is shown in Figure~\ref{fig:concrete-semantics}. Intuitively, since the collection is a \emph{array} of documents, we use standard higher-order functions for lists (e.g., \cmdMap and \cmdFilter) to formally define the semantics. \cmdDedup is the standard de-duplication function and \cmdIte is the standard if-then-else function. \cmdFlatMap is a flat map function that flattens the mapped array.

At a high level, there are seven query operators. Each query takes as input a database $\db$ and produces as output a collection. Specifically, a simple collection name $N$ just looks up the corresponding collection in the database $\db$, as defined by $\denot{N}_{\db}$.
$\cmdProject(\query, \vec{h})$ projects out specified access paths $\vec{h}$ for each document in $\denot{\query}_{\db}$.
$\cmdMatch(\query, \pred)$ filters the documents in $\denot{\query}_{\db}$ and only keeps those satisfying the predicate $\pred$.
$\cmdAddFields(\query, \vec{h}, \vec{E})$ adds new access paths $\vec{h}$ with values of expressions $\vec{E}$ to each document in $\denot{\query}_{\db}$.
$\cmdUnwind(\query, h)$ requires that the access path $h$ corresponds to an array and unwinds the result of $\query$ at $h$. In particular, it first maps each document in $\denot{\query}_{\db}$ to a list of documents where the value of $h$ in $i$-th document is the $i$-th element in the original array, and then adds all documents to the result.
$\cmdGroup(\query, \vec{k}, \vec{\attr}, \vec{A})$ groups all documents in $\denot{\query}_{\db}$ based on keys $\vec{k}$ and returns a collection that contains a document for each group with new attributes $\vec{\attr}$ and aggregated values $\vec{A}$.
$\cmdLookup(\query, k, h, N, \attr)$ adds a new attribute $\attr$ to each document in $\denot{\query}_{\db}$. The value of $\attr$ is a list of documents from the foreign collection $N$ such that the value of $k$ in $\denot{\query}_{\db}$ equals to the value of $h$ in $N$.
The semantics for predicates and expressions are straightforward from the operator names.

\section{Proofs}

\begin{proof}[Proof of Theorem~\ref{thm:abs-soundness}]
Let $\absDb$ be an abstract database over schema $\schema$, $\sketch$ be a sketch, $\query$ be a query that is a completion of $\sketch$, and $(I, O)$ be an input-output example, where $\vdash I: \schema$ and $\vdash O : \arrType{\type_O}$.
If $\denot{\query}_{I} = O$, $I \refines \absDb$, and $\absDb, \type_O \vdash \sketch \Downarrow \absSet$, then there exists an abstract collection $\absColl \in \absSet$ such that $O \refines \absColl$.
\end{proof}

\begin{proof}[Proof]
We assume that the newly generated fields (i.e. ones match \texttt{?}) attributes cannot be changed and must be all kept in the output of the whole query. We also assume that \cmdAddFields will not overwrite existing fields.

Prove by structural induction on $\sketch$. To avoid obfuscation, we use $\absDb, \type_E \vdash \sketch \Downarrow \absSet$ to denote $\sketch$ evaluates to $\absSet$ under the context of the abstract database $\absDb$ and the expected output document type $\type_E$. The context is global and will not change in evaluation.
\begin{enumerate}
    \item Base case: $\sketch = N$.
    
    Suppose that $\query = N$, $\vdash I: \schema$ and $\vdash O: \arrType{\type_O}$. By Figure \ref{fig:abstract-semantics} we have $\absDb, \type_E \vdash \sketch \Downarrow \absSet$ where $\absSet = \set{\absDb[N]}$. By Figure \ref{fig:concrete-semantics} we have $\denot{\query}_I = O$ where $O = I[N]$. By Figure \ref{fig:typing} we have $\vdash I[N]: \schema[N]$. Thus $\vdash O: \schema[N]$ and $\schema[N] = \arrType{\type_O}$. Let $\absColl = \absDb[N] \in \absSet$. Then $\absColl = (\type, l_0=|I[N]|)$ where $\schema[N] = \arrType{\type}$. Therefore $\arrType{\type_O} = \arrType{\type}$. So we have $\type_O = \type$. By the definition of match, $\type_O \matches \type$. Also the formula $l_0=|I[N]| \land l_0 = |O|$ is SAT because $O = I[N]$. Based on the above, there exists an abstract collection $\absColl = \absDb[N] \in \absSet$ s.t. $O \refines \absColl$.

    Thus, Theorem \ref{thm:abs-soundness} for the base case  $\sketch = N$ is proved.

    \item Inductive case: $\sketch' = \cmdMatch(\sketch, P)$

    Suppose that $\query' = \cmdMatch(\query, P)$, $\vdash I: \schema$, $\vdash O': \arrType{\type_O'}$, $\denot{\query'}_I = O'$, $\absDb, \type_E \vdash \sketch' \Downarrow \absSet'$.
   
    Suppose that $\vdash I: \schema$, $\vdash O: \arrType{\type_O}$, $\denot{\query}_I = O$, $\absDb, \type_E \vdash \sketch \Downarrow \absSet$. By the inductive hypothesis, there exists $\absColl \in \absSet$ s.t. $O \refines \absColl$. 

    Suppose $\absColl = (\augType, \pred)$. Then by Figure \ref{fig:abstract-semantics} we have $(\augType, \pred \land l_j \le l_i) \in\absSet'$. By Figure \ref{fig:concrete-semantics} we have $\type_O = \type_O'$, $|O'| \le |O|$. 

    By $O \refines \absColl$ we have $\type_O \matches \augType$ and $\pred \land l_i = |O|$ is SAT. Therefore $\type_O' = \type_O \matches \augType$ and $\pred \land l_j \le l_i \land l_j = |O'|$ is SAT. So there exists a $\absColl' = (\augType, \pred \land l_j \le l_i) \in \absSet'$ s.t. $O' \refines \absColl'$.
 
    Thus, Theorem \ref{thm:abs-soundness} for the inductive case $\sketch = \cmdMatch(\sketch, P)$ is proved.

    \item Inductive case: $\sketch' = \cmdProject(\sketch, \vec{h})$

    Suppose $\query' = \cmdProject(\query, \vec{h})$, $\vdash I: \schema$, $\vdash O': \arrType{\type_O'}$, $\denot{\query'}_I = O'$, $\absDb, \type_E \vdash \sketch' \Downarrow \absSet'$.

    Suppose that $\vdash I: \schema$, $\vdash O: \arrType{\type_O}$, $\denot{\query}_I = O$, $\absDb, \type_E \vdash \sketch \Downarrow \absSet$. By the inductive hypothesis, there exists $\absColl \in \absSet$ s.t. $O \refines \absColl$. 

    Suppose $\absColl = (\augType, \pred)$. Then by Figure \ref{fig:abstract-semantics} we have $((\augType - \type_k) \cup (\type_k \cap \type_E), \pred \land l_j = l_i) \in \absSet'$ where $\type_k = \cmdToDocType(\augType)$. By Figure \ref{fig:concrete-semantics} we have $\type_O' = \type_O \cap \type_E$, $|O| = |O'|$.

    By $O \refines \absColl$ we have $\type_O \matches \augType$ and $\pred \land l_i = |O|$ is SAT. Therefore $\type_k \subseteq \type_O$ and $\type_O - \type_k \matches \augType - \type_k$. Thus $\type_O \cap \type_E = (\type_k + \type_O - \type_k) \cap \type_E = (\type_k \cap \type_E) \cup ((\type_O - \type_k) \cap \type_E)$. By the assumption that newly generated fields can not be changed and must be all kept in the output, we can know $\type_O - \type_k$ is the newly generated and $\type_O - \type_k \subseteq \type_E$. So we have $((\type_O - \type_k) \cap \type_E) = (\type_O - \type_k) \matches (\augType - \type_k)$. Thus $\type_O' \matches (\augType - \type_k) \cup (\type_k \cap \type_E)$. We also have $\pred \land l_j = l_i \land l_j = |O'|$ is SAT. So there exists $\absColl' = ((\augType - \type_k) \cup (\type_k \cap \type_E), \pred \land l_j = l_i) \in \absSet'$ s.t. $O' \refines \absColl'$.
    
    Thus, Theorem \ref{thm:abs-soundness} for the inductive case $\sketch' = \cmdProject(\sketch, \vec{h})$ is proved.

    \item Inductive case: $\sketch' = \cmdAddFields(\sketch, \vec{h}, \vec{E})$

    Suppose $\query' = \cmdAddFields(\query, \vec{h}, \vec{E})$, $\vdash I: \schema$, $\vdash O': \arrType{\type_O'}$, $\denot{\query'}_I = O'$, $\absDb, \type_E \vdash \sketch' \Downarrow \absSet'$.

    Suppose that $\vdash I: \schema$, $\vdash O: \arrType{\type_O}$, $\denot{\query}_I = O$, $\absDb, \type_E \vdash \sketch \Downarrow \absSet$. By the inductive hypothesis, there exists $\absColl \in \absSet$ s.t. $O \refines \absColl$. 

    Suppose $\absColl = (\augType, \pred)$. Then by Figure \ref{fig:abstract-semantics} we have $\pair{\augType \cup $\set{$\texttt{?}_0^+$: \cmdAny}$}{\pred \land l_j = l_i} \in \absSet'$. By Figure \ref{fig:concrete-semantics} we have $|O| = |O'|$ and $\type_O \subseteq \type_O'$ where $\type_O \ne \type_O'$.

    By $O \refines \absColl$ we have $\type_O \matches \augType$ and $\pred \land l_i = |O|$ is SAT. Thus $\type_O' - \type_O \matches \set{\texttt{?}_0^+: \cmdAny}$. Therefore $\type_O' \matches \augType \cup \set{\texttt{?}_0^+: \cmdAny}$. We also have $\pred \land l_j = l_i \land l_j = |O'|$ is SAT. So there exists $\absColl' = \pair{\augType \cup $\set{$\texttt{?}_0^+$: \cmdAny}$}{\pred \land l_j = l_i} \in \absSet'$ s.t. $O' \refines \absColl'$.

    Thus, Theorem \ref{thm:abs-soundness} for the inductive case $\sketch' = \cmdAddFields(\sketch, \vec{h}, \vec{E})$ is proved.

    \item Inductive case: $\sketch' = \cmdUnwind(\sketch, h)$

    Suppose $\query' = \cmdUnwind(\query, h)$, $\vdash I: \schema$, $\vdash O': \arrType{\type_O'}$, $\denot{\query'}_I = O'$, $\absDb, \type_E \vdash \sketch' \Downarrow \absSet'$.

    Suppose that $\vdash I: \schema$, $\vdash O: \arrType{\type_O}$, $\denot{\query}_I = O$, $\absDb, \type_E \vdash \sketch \Downarrow \absSet$. By the inductive hypothesis, there exists $\absColl \in \absSet$ s.t. $O \refines \absColl$. 

    Suppose $\absColl = (\augType, \pred)$. Then by Figure \ref{fig:abstract-semantics} we have $\set{\pair{\augType[\type/a_A]}{\pred \land l_j \ge l_i} | a_A \in \augType \land \forall p.\forall q. a_A \ne \texttt{?}_p^q}  \subseteq \absSet'$ where $\textsl{Type}(a_A) = \arrType{\tau}$ and $\textsl{NotInArr}(a_A)$. By Figure \ref{fig:concrete-semantics} we have $|O'| \ge |O|$ and $\type_O' = \type_O[\type_h / a_h]$ where $\vdash a_h: \arrType{\type_h}$, $a_h \in \type_O$ and $a_h$ is fully qualified by access path $h$. 
    
    By $O \refines \absColl$ we have $\type_O \matches \augType$ and $\pred \land l_i = |O|$ is SAT. There must exist an attribute $a_A \in \augType$ such that $a_A = a_h$, and $\arrType{\type} = \arrType{\type_h}$ and $a_A$ is not a placeholder. Therefore $\type = \type_h$ and we have $\type_O[\type_h/a_h] \matches \augType[\type/a_A]$. We also have $\pred \land l_j \ge l_i \land l_j = |O'|$ is SAT. So there exists $\absColl' = \pair{\augType[\type/a_A]}{\pred \land l_j \ge l_i}$ s.t. $O' \refines \absColl'$.

    Thus, Theorem \ref{thm:abs-soundness} for the inductive case $\sketch' = \cmdUnwind(\sketch, h)$ is proved.

    \item Inductive case: $\sketch' = \cmdLookup(\sketch, h, h, N, a)$

    Suppose $\query' = \cmdLookup(\query, h, h, N, a)$, $\vdash I: \schema$, $\vdash O': \arrType{\type_O'}$, $\denot{\query'}_I = O'$, $\absDb, \type_E \vdash \sketch' \Downarrow \absSet'$.

    Suppose that $\vdash I: \schema$, $\vdash O: \arrType{\type_O}$, $\denot{\query}_I = O$, $\absDb, \type_E \vdash \sketch \Downarrow \absSet$. By the inductive hypothesis, there exists $\absColl \in \absSet$ s.t. $O \refines \absColl$. 

    Suppose $\absColl = (\augType, \pred)$. Then by Figure \ref{fig:abstract-semantics} we have $\set{\pair{\augType \cup $\set{$\texttt{?}_j^1$:$\arrType{\tau_F}$}$}{\pred \land l_j = l_i} | \tau_F \in F}  \subseteq \absSet'$ where $F = \set{\cmdToDocType(\absDb[N]_\augType) | N \in \textsl{dom}(\absDb)}$. By Figure \ref{fig:concrete-semantics} we have $|O'| = |O|$ and $\type_O' = \type_O \cup \set{\texttt{a}: \schema[N_F]}$ where $N_F \in \textsl{dom}(\schema)$.

    By $O \refines \absColl$ we have we have $\type_O \matches \augType$ and $\pred \land l_i = |O|$ is SAT. By the definition of $I \refines \absDb$ we have $\schema[N_F] = \arrType{\absDb[N_F]_\augType}$ and there are no placeholders in $\absDb[N_F]_\augType$. Then there must exist a $\type_F \in F$ such that $\type_F = \cmdToDocType(\absDb[N_F]_\augType) = \absDb[N_F]_\augType$. Thus $\schema[N_F] = \arrType{\type_F}$ and we have $\type_O \cup \set{\texttt{a}: \schema[N_F]} \matches \augType \cup $\set{$\texttt{?}_j^1$:$\arrType{\tau_F}$}$ $. We also have $\pred \land l_j = l_i \land l_j = |O'|$ is SAT. So there exists $\absColl' = \pair{\augType \cup $\set{$\texttt{?}_j^1$:$\arrType{\tau_F}$}$}{\pred \land l_j = l_i}$ s.t. $O' \refines \absColl'$.

    Thus, Theorem \ref{thm:abs-soundness} for the inductive case $\sketch' = \cmdLookup(\sketch, h, h, N, a)$ is proved.

    \item Inductive case: $\sketch' = \cmdGroup(\sketch, \vec{h}, \vec{a}, \vec{A})$

    Suppose $\query' = \cmdGroup(\query, \vec{h}, \vec{a}, \vec{A})$, $\vdash I: \schema$, $\vdash O': \arrType{\type_O'}$, $\denot{\query'}_I = O'$, $\absDb, \type_E \vdash \sketch' \Downarrow \absSet'$.

    Suppose that $\vdash I: \schema$, $\vdash O: \arrType{\type_O}$, $\denot{\query}_I = O$, $\absDb, \type_E \vdash \sketch \Downarrow \absSet$. By the inductive hypothesis, there exists $\absColl \in \absSet$ s.t. $O \refines \absColl$. 

    Suppose $\absColl = (\augType, \pred)$. Then by Figure \ref{fig:abstract-semantics} we have $\set{\pair{\texttt{\set{\_id:$\tau_K$}} \cup \type_g}{\pred \land l_j < l_i} | \tau_K \subseteq \cmdToDocType(\augType) \land \type_g \in G} \subseteq \absSet'$ where $G = $\set{\set{$\texttt{?}_j^+$: \cmdNum}, \set{}}$ $. By Figure \ref{fig:concrete-semantics} we have $|O'| < |O|$ and $\type_O' = \set{\texttt{\_id}: \set{b_1: \type_1, \ldots, b_m: \type_m}} \cup \set{a_1: \cmdNum, \ldots, a_n: \cmdNum}$ where all the access paths in $\set{b_1: \type_1, \ldots, b_m: \type_m}$ is $\vec{h}$, $\set{b_1: \type_1, \ldots, b_m: \type_m} \subseteq \type_O$, and $\vec{a} = [a_1, \ldots, a_n]$. 

    By $O \refines \absColl$ we have we have $\type_O \matches \augType$ and $\pred \land l_i = |O|$ is SAT. By the assumption that newly generated fields can not be changed and must be all kept in the output, we can know there are no placeholders in $\augType$ thus $\augType = \cmdToDocType(\augType)$ and there are newly generated fields in $\type_O$. So we have $\type_O = \augType$. There must exist a $\type_K \subseteq \cmdToDocType(\augType)$ such that $\type_K = \set{b_1: \type_1, \ldots, b_m: \type_m}$. If $n = 0$ then there exists a $\type_g = \set{}$ such that $\set{a_1: \cmdNum, \ldots, a_n: \cmdNum} \matches \type_g$ because $\set{a_1: \cmdNum, \ldots, a_n: \cmdNum}$ is an empty document type now. Otherwise if $n > 0$ the there exists a $\type_g = $\set{$\texttt{?}_j^+$: \cmdNum}$ $ such that $\set{a_1: \cmdNum, \ldots, a_n: \cmdNum} \matches \type_g$ by the definition of match. Thus we have $\set{\texttt{\_id}: \set{b_1: \type_1, \ldots, b_m: \type_m}} \cup \set{a_1: \cmdNum, \ldots, a_n: \cmdNum} \matches \texttt{\set{\_id:$\tau_K$}} \cup \type_g$. We also have $\pred \land l_j < l_i \land l_j = |O'|$ is SAT. So there exists $\absColl' = \pair{\texttt{\set{\_id:$\tau_K$}} \cup \type_g}{\pred \land l_j < l_i}$ s.t. $O' \refines \absColl'$

    Thus, Theorem \ref{thm:abs-soundness} for the inductive case $\sketch' = \cmdGroup(\sketch, \vec{h}, \vec{a}, \vec{A})$ is proved.

\end{enumerate}
\end{proof}

\begin{proof}[Proof of Theorem~\ref{thm:deduce-soundness}]
Given a database schema $\schema$, a sketch $\sketch$, and input-output examples $\vec{\exa}$, if $\textsc{Deduce}(\schema, \sketch, \vec{\exa})$ returns $\bot$, then there is no completion $\query$ of $\sketch$ such that for all $(I, O) \in \vec{\exa}$, $\denot{\query}_I = O$.
\end{proof}

\begin{proof}[Proof]
Prove by contradiction. Suppose if $\textsc{Deduce}(\schema, \sketch, \vec{\exa})$ returns $\bot$, then there is a completion $\query$ of $\sketch$ such that for all $(I, O) \in \vec{\exa}$, $\denot{\query}_I = O$. Therefore by Theorem \ref{thm:abs-soundness}, we can know for all example $\exa_j = (I_j, O_j) \in \vec{\exa}$, there exists an abstract collection $\absColl \in \absSet_j$ such that $O_j \refines \absColl$. Then by the procedure \textsc{Deduce}, we can know $\textsc{Deduce}(\schema, \sketch, \vec{\exa})$ returns $\top$. Thus there is a contradiction. So Theorem \ref{thm:deduce-soundness} is proved.
\end{proof}

\begin{proof}[Proof of Theorem~\ref{thm:soundness} (Soundness)]
Let $\schema$ be a database schema, $\vec{\exa}$ be input-output examples, and $N$ be a collection name. Suppose \textsc{CompleteSketch} is sound, if $\textsc{Synthesize}(\schema, \vec{\exa}, N)$ returns a query $\query$, then $\query$ satisfies examples $\vec{\exa}$.
\end{proof}

\begin{proof}[Proof]
By the procedure \textsc{Synthesize} we can know that if $\textsc{Synthesize}(\schema, \vec{\exa}, N)$ returns a query $\query$, then $\query \ne \bot$. By the soundness of \textsc{CompleteSketch} we have $\query$ satisfies examples $\vec{\exa}$. And there must exist a sketch $\sketch$ such that $\query$ is a completion of $\sketch$. By Theorem \ref{thm:deduce-soundness}, $\textsc{Deduce}(\schema, \sketch, \vec{\exa})$ return $\top$. Thus this $\query$ can be really returned. So Theorem \ref{thm:soundness} (Soundness) is proved.
\end{proof}

\begin{proof}[Proof of Theorem~\ref{thm:completeness} (Completeness)]
Let $\schema$ be a database schema, $\vec{\exa}$ be input-output examples, and $N$ be a collection name. Suppose \textsc{CompleteSketch} is complete, if there exists a query accepted by the grammar in Figure~\ref{fig:syntax} that is over collection $N$ and satisfies examples $\vec{\exa}$, then $\textsc{Synthesize}(\schema, \vec{\exa}, N)$ does not return $\bot$.
\end{proof}

\begin{proof}
If there exists a query $\query$ accepted by the grammar in Figure~\ref{fig:syntax} that is over collection $N$ and satisfies examples $\vec{\exa}$, then $\query$ must be a completion of a sketch $\sketch$. By Theorem \ref{thm:deduce-soundness}, $\textsc{Deduce}(\schema, \sketch, \vec{\exa})$ returns $\top$, thus $\query = \textsc{CompleteSketch}(\schema, \sketch, \vec{\exa})$. By the completeness of \textsc{CompleteSketch}, $\query \ne \bot$. Thus $\query$ will be returned, which means that $\textsc{Synthesize}(\schema, \vec{\exa}, N)$ does not return $\bot$. So Theorem \ref{thm:completeness} (Completeness) is proved.
\end{proof}



\end{document}